\documentclass{article}

\usepackage{amsmath}
\usepackage{amsfonts}
\usepackage{amsthm}
\usepackage{tikz}
\usetikzlibrary{patterns}
\usepackage[shortlabels]{enumitem}
\usepackage{colortbl}
\usepackage{multirow}
\usepackage{makecell}
\usepackage{booktabs}
\usepackage{hhline}
\usepackage[capitalise,noabbrev]{cleveref}

\usepackage{theapa}
\usepackage{relsize}    % for \mathlarger to work

\DeclareMathSymbol{\shortminus}{\mathbin}{AMSa}{"39}

\newtheorem{theorem}{\textbf{Theorem}}
\newtheorem{lemma}[theorem]{\textbf{Lemma}}
\newtheorem{corollary}[theorem]{\textbf{Corollary}}

\theoremstyle{definition}
\newtheorem{definition}{\textbf{Definition}}
\newtheorem{example}{\textbf{Example}}

\tikzset{     
  main node/.style={circle,draw,minimum size=0.75cm,inner sep=0}, 
  color1/.style={preaction={fill, red!15}}, 
  color2/.style={preaction={fill, blue!15}, pattern=north east lines, pattern color=gray}, 
  color3/.style={preaction={fill, green!15}, pattern=dots, pattern color=gray!60}, 
  color4/.style={preaction={fill, yellow!25}, pattern=crosshatch, pattern color=gray!60}, 
  color5/.style={preaction={fill, white}, pattern=grid, pattern color=gray!60}, 
}
\title{Complexity of Computing the Shapley Value in Games with Externalities}

\author{Oskar Skibski \\
\normalsize University of Warsaw, Poland}

\begin{document}

\maketitle

\begin{abstract}
\noindent 
We study the complexity of computing the Shapley value in partition function form games.
We focus on two representations based on marginal contribution nets (embedded MC-nets and weighted MC-nets) and five extensions of the Shapley value.
Our results show that while weighted MC-nets are more concise than embedded MC-nets, they have slightly worse computational properties when it comes to computing the Shapley value: two out of five extensions can be computed in polynomial time for embedded MC-nets and only one for weighted MC-nets. 
\end{abstract}

\section{Introduction}
Coalitional games are a standard model of cooperation in multi-agent systems~\cite{Chalkiadakis:etal:2011}.
In the classic form the profit of a coalition is assumed to be independent of the coalitions formed by the other players.
However, this simplifying assumption does not hold in many settings.
For example, if agents in a system have conflicted goals or limited resources, then coalitions naturally affect each other~\cite{Dunne:2005}.
There are also examples of externalities in economics, in oligopolistic markets in particular, where cooperation of some companies affect the profits of the competitors~\cite{Yi:2003}.
That is why, in the last decade \emph{coalitional games with externalities} have gained attention both in economic~\cite{Koczy:2018,Abe:Funaki:2017} and AI literature~\cite{Rahwan:etal:2012:anytime,Michalak:etal:2010:embedded}.

Externalities present new challenges both conceptually and computationally.
On the conceptual side, it is unclear how to extend most solution concepts to games with externalities.
In particular, there are several non-equivalent well-established methods of extending the Shapley value to games with externalities proposed by \citeA{PhamDo:Norde:2007} \emph{(EF-value)}, \citeA{McQuillin:2009} \emph{(MQ-value)}, \citeA{Hu:Yang:2010} \emph{(HY-value)}, \citeA{Feldman:1996} \emph{(SS-value)} and \citeA{Myerson:1977:pfg} \emph{(MY-value)}.
On the computational side, externalities highly increase the size of the game itself. 

To cope with the extensive space requirement of games with externalities, three different representations were proposed in the literature. 
The first two, called \emph{embedded MC-nets}~\cite{Michalak:etal:2010:embedded} and \emph{weighted MC-nets}~\cite{Michalak:etal:2010:weighted}, are extensions of the well-known logic-based representation: \emph{marginal contribution nets} \cite{Ieong:Shoham:2005}.
In marginal contribution nets, a game is represented as a set of rules of the form \emph{pattern $\rightarrow$ weight}. 
A coalition satisfies a rule if it fits the pattern.
Now, the value of a coalition is the sum of weights of rules it satisfies.
Embedded MC-nets and weighted MC-nets extend this formalism in a way that pattern apply not only to the coalition, but also to the partition it is embedded in.
For embedded MC-nets, \citeA{Michalak:etal:2010:embedded} proved that one extension of the Shapley value can be computed in polynomial time.
For weighted MC-nets, only partial results (polynomial results under restrictive additional assumptions) have been developed for three extensions (see \cref{section:relatedwork} for details).

More recently, \citeA{Skibski:etal:2020:pdt} proposed a new representation, named \emph{partition decision trees}.
This new representation inspired by the decision diagrams is less concise than embedded and weighted MC-nets, but as it turns out it has good computational properties.
Specifically, the authors showed that all five extensions of the Shapley value listed above can be computed in polynomial time under this representation.
So far, however, it was unknown whether the same result holds also for two other representations or not.

In this paper, we fill a gap in the literature by determining what is the complexity of computing all the five extensions of the Shapley value in games represented as embedded and weighted MC-nets.
Specifically, we show that only two out of five extensions can be computed in polynomial time for embedded MC-nets and only one can be computed in polynomial time for weighted MC-nets (unless P = NP). 
For all other values we show that computation is \#P-hard (see Table~\ref{table:summaryofresults}).

\begin{table}[t]
\centering
\begin{tabular}{l|c|c}
& Embedded MC-nets & Weighted MC-nets \\
\hline
MQ-value & in P ($^*$) & in P (Th.\ref{theorem:mq_weighted}) \\
EF-value & in P (Th.\ref{theorem:ef_embedded}) 		  	& \#P-hard (Th.\ref{theorem:ef_weighted}) \\
HY-value & \#P-hard (Th.\ref{theorem:hy_embedded}) 	& \#P-hard (Th.\ref{theorem:hy_embedded}) \\
SS-value & \#P-hard (Th.\ref{theorem:ss_embedded}) 	& \#P-hard (Th.\ref{theorem:ss_embedded}) \\
MY-value & \#P-hard (Th.\ref{theorem:my_embedded}) 	& \#P-hard (Th.\ref{theorem:my_embedded})
\end{tabular}
\caption{Summary of complexity results for computing extended Shapley value in games represented as embedded and weighted MC-nets. 
$(^*)$ Proved by \citeA{Michalak:etal:2010:embedded}.}	
\label{table:summaryofresults}
\end{table}

Our results are based on a common technique that maps embedded and weighted MC-nets into graphs. 
Specifically, first we define \emph{hybrid rules} which are weighted MC-nets rules similar in their construction to embedded MC-nets rules. 
Then, we show that every weighted MC-nets rule is equivalent to a (polynomial in size) set of hybrid rules and every embedded MC-nets rule is equivalent to some hybrid rule. 
This allows us to concentrate entirely on hybrid rules.

Furthermore, we show that every hybrid rule can be represented as a graph in which nodes are labeled with sets of players and edges indicates that some groups of players cannot be merged.
As a result, the game represented as a hybrid rule can be defined based on proper vertex colorings of the corresponding graph and every extended Shapley value is a weighted sum over all colorings in this graph.

Building upon these general results, for each value we separately analyze the resulting weighted sum.
In particular, we show the MQ-value is a weighted sum over 2-colorings, so it can be computed in polynomial time.
In turn, the EF-value is a sum over independent sets in a part of the graph and it is hard to compute unless the graph has a regular structure.
Interestingly, the SS-value under some assumptions is proved to be equal to the number of matchings in a bipartite graph, hence it is also \#P-hard to compute.

The remainder of the paper is structured as follows. 
\cref{section:relatedwork} discusses the related work.
\cref{section:preliminaries} introduces the necessary definitions and notation. 
\cref{section:from_mcnets_to_graphs} presents hybrid rules and a technique that allows us to interpret embedded and weighted MC-nets rules as a collection of graphs.
\cref{section:computing} builds upon the previous section and presents our main complexity results.
\cref{section:conclusions} concludes the paper and presents some potential future directions. 

\section{Related work}\label{section:relatedwork}

The topic of succinct representations of coalitional games has been actively studied in the last decades. 
An essential criterion in the evaluation of representations is that of efficiently computing solution concepts. 
Here, most research concentrates on the Shapley value and core-related questions: checking whether the core is empty or whether a given imputation is in the core.

There are several important representations for games without externalities. 
Arguably, \emph{weighted voting games} are the most important representation that enables to model simple 0-1 games.
In weighted voting games, a game is represented as a list of weights and a quota; the value of a coalition is $1$ if the total weight of players is equal or higher than the quota.
\citeA{Prasad:Kelly:1990} proved that checking whether the Shapley value is non-zero is NP-complete.
Later on, \citeA{Deng:Papadimitriou:1994} proved that computing the Shapley value is \#P-complete, but checking emptiness of the core is possible in polynomial time.
Since then, a lot of work has been devoted to the computational analysis of weighted voting games \cite{Matsui:Matsui:2000,Elkind:etal:2009:wvg,Zuckerman:etal:2012}.

\citeA{Deng:Papadimitriou:1994} proposed also a new representation, called \emph{induced subgraph games}.
Here, a game is represented as a graph in which nodes are players and the value of a coalition is the sum of weights of edges in the subgraph induced by the coalition. 
This representation allows to compute the Shapley value in polynomial time, but core-related questions (checking emptiness of the core, checking whether an imputation is in the core) are NP-complete.

Yet another representation only for superadditive games was proposed by \citeA{Conitzer:Sandholm:2006} under the name \emph{synergy coalition groups}. 
The conciseness comes from the fact that the values of only some coalitions are explicitly specified. 
Now, the value of a coalition is the maximal value that can be obtained from partitioning it into coalitions with specified values. 
The authors proved that checking whether an imputation is in the core can be done in polynomial time. However, computing the Shapley value or even getting the value of a coalition is computationally hard.

\emph{Marginal contribution nets}~\cite{Ieong:Shoham:2005}, extensions of which we focus on, is the first fully expressive representation.
The authors showed that for this representation the Shapley value can be computed in polynomial time, but core-related problems are computationally hard. 
This follows from the fact that marginal contribution nets without negative literals can be considered a generalization of induced subgraph games to hypergraphs.
Other complexity results for this and a richer version of the MC-nets representation were obtained by \citeA{Elkind:etal:2009:mcnets} and \citeA{Greco:etal:2011}.

Other representations for games without externalities include \emph{skill-based representations}~\cite{Ohta:etal:2009} and \emph{algebraic decision diagrams}~\cite{Aadithya:etal:2011,Ichimura:etal:2011}.

For games with externalities, however, only three representations were proposed.
The first two, \emph{embedded MC-nets}~\cite{Michalak:etal:2010:embedded} and \emph{weighted MC-nets}~\cite{Michalak:etal:2010:weighted}, are the topic of this paper.
Both representations extend the standard marginal contribution nets rules in a way that they not only specify conditions on the coalition in question, but also on the partition it is embedded in.
For embedded MC-nets, the authors proved that the MQ-value can be computed in polynomial time.
For weighted MC-nets, only partial results for the MQ-value, the EF-value and the MY-value were obtained.
Specifically, as we describe in the next section, a weighted MC-nets rule consists of several blocks, each containing one or more (standard) marginal contribution nets rules:
\[ (pattern^1_1 \rightarrow value^1_1) \dots (pattern^1_k \rightarrow value^1_k) \mid (pattern_1^2 \rightarrow value_1^2) \ldots \mid \dots. \]
Now, the authors designed polynomial algorithms that work under additional assumptions. 
For the MQ-value, it is assumed that there exist some patterns that do not apply to singleton coalitions (with more than one positive literal) and, if they are in the same block, then other rules in this block are pairwise compatible.
For the EF-value, the algorithms assumes that there is more than one block or one block with one rule.
For the MY-value, the algorithms assumes that there is only one block and all rules in it are compatible.
See \cite{Michalak:etal:2010:weighted} and \cite{Michalak:2016:erratum} for details.
As we show in our paper, without these restrictions, it is indeed possible to compute the MQ-value in polynomial time, but computing the EF-value and the MY-value is \#P-hard.

In the third representation, named \emph{partition decision trees}~\cite{Skibski:etal:2020:pdt}, one rule is a directed tree, where non-leaf nodes are labeled with players' names, leaf nodes are labeled with payoff vectors, and edges indicate membership of players in coalitions. 
In this way, paths of the partition decision trees are equivalent to simple weighted MC-nets rules (more precisely, to weighted MC-nets without negative literals in which each block contains exactly one rule).
The authors proved that for this representation the MQ-value, the EF-value, the HY-value, the SS-value and the MY-value can all be computed in polynomial time.
This representation, however, is much less concise than embedded and weighted MC-nets.

Our results make use of the graphs that are labeled with sets of players (we call them \emph{player-graphs}, see Section~\ref{section:player-graphs}). 
In particular, if each set contains exactly one player we obtain a graph in which nodes are players and edges represent restrictions in forming a coalition. 
That is why, our paper belongs to a line of work that concentrate on computing solution concepts, and the Shapley value in particular, for classes of games defined based on graphs.
These include matching games~\cite{Greco:etal:2020,Aziz:DeKeijzer:2011}, network flow games~\cite{Bachrach:Rosenschein:2009} and game-theoretic centralities~\cite{Michalak:etal:2013:efficient,Skibski:etal:2019:attachment}, among others.

Moreover, the idea of player-graphs is also similar to Myerson's graph-restricted games \cite{Myerson:1977} and even more to compatibility games \cite{See:etal:2014} or signed graph games \cite{Skibski:etal:2020:signed} as an edge represent that two players cannot be in the same coalition.
The main difference here is that the graph in this paper (along with the weight of a rule) is the source of the values in a game and not an additional structure that affects values of the existing game.
Another difference is the fact that these models are proposed for games without externalities, and we study games with externalities.

Last but not least, \citeA{Skibski:etal:2016:kcoalitional} proposed a model of games with externalities in which at most $k$ coalitions can form. 
They proposed a dedicated extension of the Shapley value and proved it is \#P-hard to compute if the game is represented as a single embedded MC-nets rule. 
As far as we know, this is the only paper---other than already mentioned papers that introduce embedded MC-nets, weighted MC-nets and partition decision trees---that analyzes the complexity of extended Shapley values.

\section{Preliminaries}\label{section:preliminaries}
In this section, we introduce basic notation and definitions.
In many definitions we will use Iverson brackets: $[\varphi] = 1$ if statement $\varphi$ is true, and $[\varphi]=0$, otherwise (e.g., $[2=1]=0$).

\subsection{Partition function form games}

Let $N = \{1,\dots,n\}$ be a set of $n$ players, which will be fixed throughout the paper.
A \emph{coalition} is any nonempty subset of $N$.
The set of all possible partition of $N$ is denoted by $\mathcal{P}$ and the set of all embedded coalitions, i.e., coalitions in partitions, by $EC$:
\[ EC = \{(S,P) : P \in \mathcal{P}, S \in P\}. \]

In this paper, by \emph{game} we mean a coalitional game with externalities in a partition function form: formally, for a fixed set of players, a \emph{game} is a function that assigns a real value to every embedded coalition: $g: EC \rightarrow \mathbb{R}$. 
We say that a game has no externalities if the value of every coalition does not depend on the partition, i.e., $g(S,P) = g(S,P')$ for every coalition $S \subseteq N$ and $(S,P), (S,P') \in EC$.
If game has no externalities, then there exists a function $\hat{g}: 2^N \rightarrow \mathbb{R}$ with $\hat{g}(\emptyset) = 0$ such that $g(S,P) = \hat{g}(S)$ for every $(S,P) \in EC$.
This function is called a \emph{characteristic function}.

\subsection{Extended Shapley values}
A \emph{value} of a player in a game is a real number that represents player's importance or expected outcome.
The \emph{Shapley value} \cite{Shapley:1953} is defined for games without externalities.
For a game without externalities defined through a characteristic function $\hat{g}$, the Shapley value of player $i \in N$ is defined as follows:
\[ SV_i(\hat{g}) = \sum_{S \subseteq N} \zeta_i(S) \cdot \hat{g}(S), \mbox{ \quad where \ } \zeta_i(S) = \begin{cases} 
\frac{(|S|-1)! (n-|S|)!}{n!} & \mbox{if } i \in S, \\ 
-\frac{|S|! (n-|S|-1)!}{n!}  & \mbox{otherwise.}
\end{cases}\]
The Shapley value can be interpreted in the following way.
Assume players leave the grand coalition $N$ one by one in a random order. 
As the player leaves, she receives the payoff equal to the marginal contribution to the group of players $S$ they left: $\hat{g}(S) - \hat{g}(S \setminus \{i\})$.
Now, the Shapley value is expected payoff over all orders.

\citeA{Shapley:1953} famously proved that the Shapley value is a unique value that satisfies four simple axioms: \emph{Efficiency}, \emph{Symmetry}, \emph{Additivity} and \emph{Null-Player}.
However, in games with externalities these classic axioms are too weak to guarantee uniqueness~\cite{Macho-Stadler:etal:2007}. 
Hence, numerous extensions of the Shapley value to games with externalities have been proposed (see \cite{Koczy:2018} for a recent overview). 
Most of them, and all that we will focus on, can be considered a weighted average of values of embedded coalitions:
\begin{equation}\label{eq:extended_sv_general_form}
\varphi_i(g) = \sum_{(S,P) \in EC} \omega_i(S,P) \cdot g(S,P)
\end{equation}
for some weights $\omega: EC \times N \rightarrow \mathbb{R}$.

There are six extended Shapley values that satisfies the standard translation of the Shapley's axioms.
In this paper, we will focus on five of them, leaving the problem open for the Bolger value~\cite{Bolger:1989}.
To date, no computational results for the Bolger value are known.
In particular, the Bolger value satisfies \cref{eq:extended_sv_general_form}, but even computing weights $\omega_i(S,P)$ seems to be computationally challenging.
It is also not known whether the Bolger value can be computed in polynomial time in games represented as partition decision trees~\cite{Skibski:etal:2020:pdt}.

Let us describe the remaining five extensions.
Recall \cref{eq:extended_sv_general_form}.
In order to give the same values as the (standard) Shapley value for games without externalities, for a fixed coalition $S$ weights of embedded coalitions $(S,P) \in EC$ should sum up to $\zeta_i(S)$:
$\sum_{P \ni S} \omega_i(S,P) = \zeta_i(S)$.
Hence, such extended Shapley values can be defined by specifying how weight $\zeta_i(S)$ is distributed among partitions that contain $S$.

In the simplest approach, the whole weight is assigned to one partition.
This is the case for the first two values:
\begin{description}
\item[McQuillin value (MQ-value)] \cite{McQuillin:2009} of player $i$ in game $g$ is defined as
\begin{equation}\label{eq:mq_value}
MQ_i(g) = \sum_{(S,P) \in EC} \zeta_i(S) [|P| \le 2] \cdot g(S,P).
\end{equation}
\item[Externality-free value (EF-value)]\cite{PhamDo:Norde:2007} of player $i$ in game $g$ is defined as
\begin{equation}\label{eq:ef_value}
EF_i(g) = \sum_{(S,P) \in EC} \zeta_i(S) [|P|-1 = n - |S|] \cdot g(S,P).
\end{equation}
\end{description}
The MQ-value looks only at the value of coalition $S$ in a partition in which all other players form one coalition.
In turn, the EF-value (independently proposed by \citeA{DeClippel:Serrano:2008}) looks only at the value of coalition $S$ in a partition in which all other players form singleton coalitions.
As a result, MQ-value and EF-value are equivalent to the Shapley value of a game without externalities defined for every coalition $S \subseteq N$ as $\hat{g}(S) = g(S, \{S,N\setminus S\})$ in the case of MQ-value and $\hat{g}(S) = g(S, \{S\} \cup \{\{i\} : i \in N \setminus S\})$ in the case of EF-value.

The MQ-value and the EF-value ignore values of most embedded coalitions.
In turn, the HY-value and the SS-value for every coalition take a weighted average over all partitions it is in; the difference is the HY-value assigns greater weights to partitions with more coalitions, while the SS-value---to partitions with larger coalitions:
\begin{description}
\item [Hu-Yang value (HY-value)] \cite{Hu:Yang:2010} of player $i$ in game $g$ is defined as 
\begin{equation}\label{eq:hy_value}
HY_i(g) = \sum_{(S,P) \in EC} \zeta_i(S) \frac{\psi(S,P)}{|P|} \cdot g(S,P),
\end{equation}
where $\psi(S,P)$ is the number of partitions of $N$ in which players $N \setminus S$ form partition $P \setminus \{S\}$: $\psi(S,P) = |\{P' \in \mathcal{P} : \{T \setminus S : T \in P'\} \setminus \{\emptyset\} = P \setminus \{S\}\}|$.

\item[Stochastic Shapley value (SS-value)] \cite{Feldman:1996} of player $i$ in game $g$ is defined as
\begin{equation}\label{eq:ss_value}
SS_i(g) = \sum_{(S,P) \in EC} \zeta_i(S) \frac{\prod_{T \in P \setminus \{S\}}(|T|-1)!}{(n-|S|)!} \cdot g(S,P).
\end{equation}
\end{description}

For the HY-value, the value $\psi(S,P)$ is equal to the number of possible partitions that can be obtained by inserting players $S$ into partition $P \setminus \{S\}$.
Hence, for a fixed $S$, $\psi(S,P)$ depends solely on the size of partition $P$ and is larger for larger $P$.

For the SS-value, first proposed by \citeA{Feldman:1996}, but better known from the work of \citeA{Macho-Stadler:etal:2007}, the fraction ${\prod_{T \in P \setminus \{S\}}(|T|-1)!}/{(n-|S|)!}$ is the probability that in a random permutation of players $N \setminus S$ each coalition from $P \setminus \{S\}$ will form a separate cycle in a cycle decomposition.
As a result, weights of small partitions are significantly larger than weights of large partitions (e.g., $g(S,\{S,N\setminus S\})$ is multiplied by $\zeta_i(S)/(n-|S|)$ and $g(S,\{S,\{\{j\} : j \in N \setminus S\}\})$ by $\zeta_i(S)/(n-|S|)!$).

Finally, \citeA{Myerson:1977:pfg} studied how the Null-player axiom can be strengthen in order to obtain unique characterization. 
His analysis based on inclusion-exclusion principle led to the following value:
\begin{description}
\item[Myerson value (MY-value)] \cite{Myerson:1977:pfg} of player $i$ in game $g$ is defined as
\begin{equation}\label{eq:my_value}
MY_i(g) = \sum_{(S,P) \in EC} (-1)^{|P|} \left( \left(\sum_{T \in P \setminus \{S\}, i \not \in T} \frac{(|P|-2)!}{(n-|T|)}\right) - \frac{(|P|-1)!}{n}\right) g(S,P).
\end{equation}
\end{description}
What is characteristic for the MY-value, the weight of an embedded coalition $(S,P)$ may be negative even if $i \in S$. 
This implies that the MY-value does not satisfy a basic monotonicity principle: increasing the value of $g(S,P)$ may decrease the value of player $i \in S$.

\begin{example}
Fix $N = \{1,\dots,6\}$, coalition $S = \{1,2\}$ and player $i = 1$. 
Consider an elementary game without externalities $e^S$ in which only coalition $S$ has non-zero value: $e^S(S) = 1$ and $e^S(T) = 0$ for $S \neq T$.
The Shapley value of player $i$ in this game equals $SV_i(e^S) = \zeta_i(S) = 1/30$.

Now, consider elementary games with externalities.
For the fixed coalition $S=\{1,2\}$, there are 15 embedded coalitions $(S,P) \in EC$.
For each, we consider an elementary game $e^{(S,P)}$ in which only $(S,P)$ has non-zero value: $e^{(S,P)}(S,P) = 1$ and $e^{(S,P)}(T,R) = 0$ for $(T,R) \neq (S,P)$.
Note that $\sum_{P \ni S} e^{(S,P)}$ is equivalent to $e^S$, hence for every considered extended Shapley value $\varphi$ we have $\varphi_i(\sum_{P \ni S} e^{(S,P)}) = \zeta_i(S) = 1/30$. 

Table~\ref{table:esv} presents how value $\zeta_i(S)$ is distributed among elementary games $e^{(S,P)}$ by extended Shapley values. 
Specifically, it contains values $MQ_i(g)$, $EF_i(g)$, $HY_i(g)$, $SS_i(g)$ and $MY_i(g)$ divided by $\zeta_i(S)$ for games $g \in \{e^{(S,P)} : (S,P) \in EC\}$.
%These values corresponds also to the weights of values $g(S,P)$ in the corresponding formulas of the extended Shapley values.
For example, we have $HY_i(e^{(S,P)}) = \zeta_i(S) \cdot 17/203$ for $(S,P) = (\{1,2\}, \{\{1,2\}, \{3,4\}, \{5\}, \{6\}\})$.
\end{example}

\begin{table}
\centering
\setlength{\tabcolsep}{10pt}
\begin{tabular}{cc|ccccc}
& & \multicolumn{5}{c}{$\varphi_i(e^{(S,P)})/\zeta_i(S)$}\\
$S$ & $P \setminus \{S\}$ & $MQ_i$ & $EF_i$ & $HY_i$ & $SS_i$ & $MY_i$ \\
\hline
$\{1,2\}$ & $\{\{3,4,5,6\}\}$ 	
& 1 & 0 & 5/203 & 6/24 & 10 \\
\hline
$\{1,2\}$ & $\{\{3,4,5\}, \{6\}\}$ 
& 0 & 0 & 10/203 & 2/24 & -6 \\
$\{1,2\}$ & $\{\{3,4,6\}, \{5\}\}$ 
& 0 & 0 & 10/203 & 2/24 & -6 \\
$\{1,2\}$ & $\{\{3,5,6\}, \{4\}\}$ 
& 0 & 0 & 10/203 & 2/24 & -6 \\
$\{1,2\}$ & $\{\{3\}, \{4,5,6\}\}$ 
& 0 & 0 & 10/203 & 2/24 & -6 \\
\hline
$\{1,2\}$ & $\{\{3,4\}, \{5,6\}\}$ 
& 0 & 0 & 10/203 & 1/24 & -5 \\
$\{1,2\}$ & $\{\{3,5\}, \{4,6\}\}$ 
& 0 & 0 & 10/203 & 1/24 & -5 \\
$\{1,2\}$ & $\{\{3,6\}, \{4,5\}\}$ 
& 0 & 0 & 10/203 & 1/24 & -5 \\
\hline
$\{1,2\}$ & $\{\{3,4\}, \{5\}, \{6\}\}$ 
& 0 & 0 & 17/203 & 1/24 & 9 \\
$\{1,2\}$ & $\{\{3,5\}, \{4\}, \{6\}\}$ 
& 0 & 0 & 17/203 & 1/24 & 9 \\
$\{1,2\}$ & $\{\{3,6\}, \{4\}, \{5\}\}$ 
& 0 & 0 & 17/203 & 1/24 & 9 \\
$\{1,2\}$ & $\{\{3\}, \{4,5\}, \{6\}\}$ 
& 0 & 0 & 17/203 & 1/24 & 9 \\
$\{1,2\}$ & $\{\{3\}, \{4,6\}, \{5\}\}$ 
& 0 & 0 & 17/203 & 1/24 & 9 \\
$\{1,2\}$ & $\{\{3\}, \{4\}, \{5,6\}\}$ 
& 0 & 0 & 17/203 & 1/24 & 9 \\
\hline
$\{1,2\}$ & $\{\{3\},\{4\},\{5\},\{6\}\}$ 
& 0 & 1 & 26/203 & 1/24 & -24 \\
\end{tabular}
\caption{Values $\varphi_i(e^{(S,P)})/\zeta_i(S)$ for $N \!=\! \{1,\dots,6\}$, $S = \{1,2\}$, $i=1$, arbitrary partition $P$ that contains $S$ and different extended Shapley values $\varphi$.
%Cells most important in our complexity results are marked.
}
\label{table:esv}
\end{table}

Since extended Shapley values are the main subject of our work, let us provide some additional intuition behind them.
Every extended Shapley value that satisfies Shapley's axioms can be obtained using the \emph{process approach}, similar to the interpretation of the Shapley value \cite{Skibski:etal:2018:stochastic}.
Assume players leave the grand coalition one by one in a random order and divide themselves into groups outside.
As the player leaves, she chooses one of the coalitions outside or creates a new coalition, each with some probability. 
As a result, the player receives the payoff equal to the change in the value of the coalition of players they left.
Specifically, if player $i$ leaves coalition $S$ ($i \in S$) in partition $P$ ($S \in P$) which results in forming $S' = S \setminus \{i\}$ in partition $P'$, then the payoff equals $g(S,P) - g(S',P')$.
Now, the extended Shapley value is the expected payoff in this process.

Extended Shapley values differ in probabilities of choices made by the leaving player:
\begin{itemize}
\item According to the MQ-value, the player always joins the (unique) existing coalition outside.
\item According to the EF-value, the player always creates a new coalition.
\item According to the HY-value, each final partition should have the same probability; hence, the probability is $\psi(S',P')/\psi(S,P)$.
\item According to the SS-value, the probability of joining a coalition is proportional to its size; specifically, it joins coalition of size $T$ with probability $|T|/(|N|-|S|+1)$ and creates a new one with probability $1/(|N|-|S|+1)$.
\item The MY-value can also be obtained using the process approach, but more complex quasi-probabilities instead of probabilities need to be used.
\end{itemize}

Figure~\ref{fig:process_approach} illustrates the process approach.
For the MQ-value the probability of forming grand coalition outside equals $1$; hence, $(\alpha_1,\alpha_2) = (1,0)$ and $(\beta_1,\beta_2) = (1,0)$.
In turn, for the EF-value the probability of forming partition $\{\{1\},\{2\},\{3\}\}$ equals $1$; hence, $(\alpha_1,\alpha_2)=(0,1)$ and $(\gamma_1,\gamma_2,\gamma_3)=(0,0,1)$.
For the HY-value, the probability of each final partition is equal; hence, $(\alpha_1,\alpha_2)=(2/5,3/5)$, $(\beta_1,\beta_2)=(1/2,1/2)$ and $(\gamma_1,\gamma_2,\gamma_3) = (1/3, 1/3, 1/3)$.
For the SS-value, the probability of joining a larger coalition is larger; hence, $(\alpha_1,\alpha_2) = (1/2,1/2)$, $(\beta_1,\beta_2)=(2/3,1/3)$ and $(\gamma_1,\gamma_2,\gamma_3)=(1/3,1/3,1/3)$.
In case of the MY-value, we only note that $(\alpha_1,\alpha_2)=(-1,2)$.

\begin{figure}[t]
\centering
\begin{tikzpicture}[x=15cm,y=5cm] % change these values to adjust the size of a figure
  \tikzset{     
    e4c node/.style={rectangle,draw,inner sep=2}, 
    e4c edge/.style={above,font=\footnotesize}
  }
  \node[e4c node] (1) at (0.20, 0.825) {$\{1,2,3\}\mid\emptyset$}; 
  \node[e4c node] (2) at (0.40, 0.825) {$\{2,3\}\mid\{1\}$}; 
  \node[e4c node,minimum width=2.2cm] (3) at (0.65, 1.00) {$\{3\}\mid\{1,2\}$}; 
  \node[e4c node,minimum width=2.2cm] (4) at (0.65, 0.65) {$\{3\}\mid\{1\},\{2\}$}; 
  \node[e4c node,minimum width=2.8cm] (5) at (1.00, 1.07) {$\emptyset\ \mid\{1,2,3\}$}; 
  \node[e4c node,minimum width=2.8cm] (6) at (1.00, 0.93) {$\emptyset\ \mid\{1,2\},\{3\}$}; 
  \node[e4c node,minimum width=2.8cm] (7) at (1.00, 0.79) {$\emptyset\ \mid\{1,3\},\{2\}$}; 
  \node[e4c node,minimum width=2.8cm] (8) at (1.00, 0.65) {$\emptyset\ \mid\{1\},\{2,3\}$};
  \node[e4c node,minimum width=2.8cm] (9) at (1.00, 0.51) {$\emptyset\mid\{1\},\{2\},\{3\}$};

  \path[->,draw,thin]
  (1) edge[e4c edge] node[yshift=-1] {$1$} (2)
  (2) edge[e4c edge] node[yshift=-1] {$\alpha_1$} (3)
  (2) edge[e4c edge] node[yshift=-1] {$\alpha_2$} (4)

  (3) edge[e4c edge,above,pos=0.7] node[yshift=-1] {$\beta_1$} (5)
  (3) edge[e4c edge,above,pos=0.7] node[yshift=-3] {$\beta_2$} (6)

  (4) edge[e4c edge,above,pos=0.7] node[yshift=-1] {$\gamma_1$} (7)
  (4) edge[e4c edge,above,pos=0.7] node[yshift=-3] {$\gamma_2$} (8)
  (4) edge[e4c edge,above,pos=0.7] node[yshift=-3] {$\gamma_3$} (9)
  ;
\end{tikzpicture}
\caption{An illustration of the process approach for a fixed permutation $(1,2,3)$.} 
\label{fig:process_approach}
\end{figure}

\subsection{Representations}
We will consider two representations for games with externalities.
Both are extensions of the MC-nets representation.

In all considered representations the game is represented as a set of rules: $\Gamma = \{\gamma_1,\dots,\gamma_k\}$.
The rules, however, differ between representations.
In what follows, we will denote the game represented as a single rule $\gamma$ by $g^{\gamma}$ and as the set of rules $\Gamma$ by $g^{\Gamma}$.

\emph{Marginal contribution nets (MC-nets)} \cite{Ieong:Shoham:2005} are a representation for games without externalities.
The game is represented as a set of MC-nets rules of the form: $(\alpha \rightarrow c)$.
Here, $c \in \mathbb{R}$ is the weight of a rule and $\alpha$ is a boolean expression over $N$ of the form:
\begin{equation}\label{eq:mcnet_form_alpha} 
(a^+_1 \land \dots \land a^+_m \land \lnot a^-_1 \land \dots \land \lnot a^-_l),
\end{equation}
where $a_1^+, \dots, a_m^+ \in N$ are called \emph{positive literals} and $a_1^-, \dots, a_l^- \in N$ are called \emph{negative literals}. 
We denote sets of positive and negative literals by $\oplus(\alpha)$ and $\ominus(\alpha)$, respectively, and assume $\oplus(\alpha) \cap \ominus(\alpha) = \emptyset$ and $\oplus(\alpha) \neq \emptyset$.\footnote{\citeA{Ieong:Shoham:2005} allows rules without positive literals which entails that the empty coalition may have non-zero value. As standard in the literature, we do not allow such situations.}
A coalition $S$ \emph{satisfies} $\alpha$ if it contains all positive literals and does not contain any negative literal, i.e., $\oplus(\alpha) \subseteq S$ and $\ominus(\alpha) \cap S = \emptyset$.
Now, in a game represented as a set of MC-nets rules the value of coalition $S$ is the sum of weights of all satisfied rules.

\begin{example}\label{example:mcnets}
Assume $N = \{1,2,3,4\}$ and consider the following MC-nets rules $\Gamma = \{\gamma_1, \gamma_2\}$:
\[ \gamma_1 = (1 \land 2 \rightarrow c_1), \quad
\gamma_2 = (3 \land \lnot 4 \rightarrow c_2). \]
A coalition has non-zero value if it satisfies at least one of them.
We get:
\[ \hat{g}^{\Gamma}(S) = c_1 \cdot [\{1, 2\} \subseteq S] + c_2 \cdot [3 \in S \land 4 \not \in S] = \begin{cases}
c_1 & \mbox{if } S \in \{\{1,2\}, \{1,2,4\}, \{1,2,3,4\}\}, \\
c_2 & \mbox{if } S \in \{\{3\}, \{1,3\}, \{2,3\}\}, \\
c_1+c_2 & \mbox{if } S = \{1,2,3\}, \\
0 & \mbox{otherwise. }
\end{cases} \]
\end{example}

\subsubsection{Embedded MC-nets \cite{Michalak:etal:2010:embedded}}
\emph{Embedded MC-nets} extends (standard) MC-nets as they allow to specify not only restrictions on the coalition, but also on the partition this coalition is embedded in.

An embedded MC-nets rule is of the form:
\[ (\alpha_1 \mid \alpha_2, \alpha_3, \dots, \alpha_k) \rightarrow c, \]
where $c \in \mathbb{R}$ is the weight of a rule and $\alpha_1, \dots, \alpha_k$ are boolean expressions as in \cref{eq:mcnet_form_alpha}.
An embedded coalition $(S,P)$ satisfies the rule if $S$ satisfies $\alpha_1$ and for every $\alpha_i$ with $i>1$ there exists a coalition $T \in P \setminus \{S\}$ that satisfies it.
Now, in a game represented as a set of embedded MC-nets rules the value of embedded coalition $(S,P)$ is the sum of weights of all satisfied rules.

If $k=1$, then the embedded MC-nets rule is equivalent to a (standard) MC-nets rule.

\begin{example}\label{example:embedded_mcnets}
Assume $N = \{1,\dots,6\}$ and consider the following embedded MC-nets rule:
\[ \gamma = (1 \land 2 \rightarrow 1) \mid (3 \land 5 \land \lnot 4 \land \lnot 6) \ (4 \land \lnot 6) \]
An embedded coalition $(S,P)$ satisfies this rule if $S$ contains $1$ and $2$, and $P \setminus \{S\}$ contains a coalition with $3$ and $5$, but without $4$ and $6$ and a coalition with $4$, but without $6$. 
We get that:
\[ g^{\gamma}(\{1,2\}, \{\{1,2\}, \{3,5\}, \{4\}, \{6\}\}) = 1, \quad g^{\gamma}(\{1,2,6\}, \{\{1,2,6\}, \{3,5\}, \{4\}\}) = 1 \]
and $g^{\gamma}(S,P) = 0$ otherwise.
\end{example}

\subsubsection{Weighted MC-nets \cite{Michalak:etal:2010:weighted}}
In \emph{weighted MC-nets}, one rule is a partitioned set of (standard) MC-nets rules; roughly speaking, a rule applies to an embedded coalition if all (standard) rules are satisfied at the same time.

A weighted MC-nets rule is of the form:
\[ (\alpha_1^1 \rightarrow c_1^1) \dots (\alpha_{k_1}^1 \rightarrow c_{k_1}^1) \mid \dots \mid (\alpha_1^m \rightarrow c_1^m) \dots (\alpha_{k_m}^m \rightarrow c_{k_m}^m), \]
where $(\alpha_j^i \rightarrow c_j^i)$ for every $i \in \{1,\dots,m\}$ and $j \in \{1,\dots,k_i\}$ is an MC-nets rule.
A partition $P$ satisfies the rule if it can be partitioned into $m$ disjoint subsets $P = R_1 \dot\cup \dots \dot\cup R_m$ such that  for every $i \in \{1,\dots,m\}$ and $j \in \{1,\dots,k_i\}$ rule $(\alpha_j^i \rightarrow c_j^i)$ is satisfied by some coalition from $R_i$.
Now, in a game represented as a set of weighted MC-nets rules the value of embedded coalition $(S,P)$ is the sum of weights of all MC-nets rules $(\alpha \rightarrow c)$ that $S$ satisfies in all (weighted MC-nets) rules satisfied by $P$.

If $m=2$, $k_1=1$ and $c^i_j = 0$ for $i = 2$, then the weighted MC-nets rule is equivalent to an embedded MC-nets rule.
If $m=1$ and $k_1=1$, then the weighted MC-nets rule is a standard MC-nets rule.

\begin{example}\label{example:weighted_mcnets}
Assume $N = \{1,\dots,6\}$ and consider the following weighted MC-nets rule:
\[ \gamma = (1 \land 2 \land \lnot 4 \rightarrow c_1) \ (6 \land \lnot 4 \rightarrow c_2) \mid (3 \land 5 \land \lnot 4 \rightarrow c_3)\]
Partition $P$ satisfies this rule if $P$ can be divided into two parts $P = R_1 \dot\cup R_2$ such that
\begin{itemize}
\item in $R_1$ there exists a coalition with both $1$ and $2$ and a coalition with $6$, but both does not contain $4$
\item in $R_2$ there exists a coalition with both $3$ and $5$, but without $4$.
\end{itemize}
We get that only $P_1 = \{\{1,2\}, \{3,5\}, \{4\}, \{6\}\}$ and $P_2 = \{\{1,2,6\}, \{3,5\}, \{4\}\}$ satisfies $\gamma$.
By summing weights of all MC-nets satisfied by each coalition from these partitions we get:
\begin{center}
\begin{tabular}{lll}
$g^{\gamma}(\{1,2\},P_1) = c_1$, & $g^{\gamma}(\{3,5\},P_1) = c_3,$ \ \ & $g^{\gamma}(\{6\},P_1) = c_2$, \\
$g^{\gamma}(\{1,2,6\},P_2) = c_1+c_2$, \ \  & $g^{\gamma}(\{3,5\},P_2) = c_3$
\end{tabular}
\end{center}
and $g^{\gamma}(S,P) = 0$, otherwise.
For $c_1=1$ and $c_2=c_3 = 0$ the weighted MC-nets rule is equivalent to embedded MC-nets rule from \cref{example:embedded_mcnets}.
\end{example}

\section{From MC-Nets to Graphs}\label{section:from_mcnets_to_graphs}
The goal of this section is to show that (1) every embedded and weighted MC-nets rule is equivalent to a set of \emph{hybrid rules}, (2) every hybrid rule can be represented as a graph and, as a result, (3) every game represented as embedded or weighted MC-nets rules can be defined based on (proper vertex) colorings in such graphs.

%Lemmas~\ref{lemma:hybrid_rules_weighted} and \ref{lemma:hybrid_rules_embedded} will be crucial in our complexity analysis as they allow us to focus on hybrid rules.

\subsection{Hybrid rules}
We begin by introducing a subclass of weighted MC-nets rules under the name \emph{hybrid rules}.
The name comes from the fact that hybrid rules, while they are formally weighted MC-nets rules, have a form almost identical to embedded MC-nets rules.

\begin{definition}\label{definition:simple} (Hybrid rules)
A \emph{hybrid rule} is a weighted MC-nets rule with $m=1$ of the form:
\[ \gamma = (\alpha_1 \rightarrow c) (\alpha_2 \rightarrow 0) \dots (\alpha_k \rightarrow 0), \]
with $\oplus(\alpha_i) \cap \oplus(\alpha_j) = \emptyset$ for $i,j \in \{1,\dots,k\}$, $i \neq j$ and $\bigcup_{i=1}^k \oplus(\alpha_i) = N$.
We will call $c \in \mathbb{R}$ the weight of rule $\gamma$.
\end{definition}
%We will use $k$ as the size of $\gamma$ from this point onward.

Note that for every hybrid rule $\{\oplus(\alpha_1), \dots, \oplus(\alpha_k)\}$ is a partition of $N$.
Based on the definition of weighted MC-nets, embedded coalition $(S,P)$ satisfies hybrid rule $\gamma$ if $S$ satisfies $\alpha_1$ and for every $\alpha_i$ with $i>1$ there exists a coalition in $P$ that satisfies it. Note that, unlike embedded MC-nets, $S$ may also satisfy $\alpha_i$ for $i>1$.

In \cref{lemma:hybrid_rules_weighted}, we show that every weighted MC-nets rule can be expressed using polynomially many hybrid rules. 
We will use a notion of compatibility: we say that expressions $\alpha_i, \alpha_j$ are \emph{compatible}, denoted by $\alpha_i \sim \alpha_j$, if there exists a coalition that satisfies both of them, i.e., if $(\oplus(\alpha_i) \cup \oplus(\alpha_j)) \cap (\ominus(\alpha_i) \cup \ominus(\alpha_j)) = \emptyset$.

\begin{lemma}\label{lemma:hybrid_rules_weighted}
Every weighted MC-nets rule of size $S$ is equivalent to a set of hybrid rules of size $\mathrm{poly}(n,S)$.
\end{lemma}

\begin{proof}
Consider a weighted MC-nets rule:
\[ (\alpha_1^1 \rightarrow c_1^1) \dots (\alpha_{k_1}^1 \rightarrow c_{k_1}^1) \mid \dots \mid (\alpha_1^m \rightarrow c_1^m) \dots (\alpha_{k_m}^m \rightarrow c_{k_m}^m). \]
Assume $\oplus(\alpha_j^i) \cap \oplus(\alpha_{j'}^{i'}) \neq \emptyset$ for some $i,i' \in \{1,\dots,m\}$, $j \in \{1,\dots,k_i\}$, $j' \in \{1,\dots,k_{i'}\}$.
If $i \neq i'$ or $i=i'$ but expressions $\alpha_j^i$ and $\alpha_{j'}^{i'}$ are not compatible, then the weighted MC-nets rule is contradictory; hence it is equivalent to an empty set of hybrid rules.
If $i = i'$ and $\alpha_j^i$ and $\alpha_{j'}^{i'}$ are compatible, then rules $(\alpha_j^i \rightarrow c_j^i)$ and $(\alpha_{j'}^{i'} \rightarrow c_{j'}^{i'})$ can be combined into $(\alpha_j^i \land \alpha_{j'}^{i'} \rightarrow (c_j^i + c_{j'}^{i'}))$. 
Hence, in what follows, we assume that $\oplus(\alpha_j^i) \cap \oplus(\alpha_{j'}^{i'}) = \emptyset$ for every $i,i' \in \{1,\dots,m\}$, $j \in \{1,\dots,k_i\}$, $j' \in \{1,\dots,k_{i'}\}$.

Define $\oplus^i = \bigcup_{j=1}^{k_i} \oplus(\alpha_j^i)$ for every $i \in \{1,\dots,m\}$ and $\oplus = \bigcup_{i=1}^m \oplus^i$.
Now, define $\beta_j^i$ as $\alpha_j^i$ with players $\oplus \setminus \oplus^i$ added as negative literals for every $i \in \{1,\dots,m\}$ and $j \in \{1,\dots, k_i\}$. Since $\beta_j^i$, $\beta_{j'}^{i'}$ for $i \neq i'$ cannot be satisfied by the same coalition, we get that the following weighted MC-nets rule without bars ("$|$") is equivalent to the original one:
\[ (\beta_1^1 \rightarrow c_1^1) \dots (\beta_j^i \rightarrow c_j^i) \dots (\beta_{k_m}^m \rightarrow c_{k_m}^m). \]
In each MC-nets rule we added at most $S$ literals, so the size of this rule is $O(S^2)$.

Now, we divide the formula into $k_1 + \dots + k_m$ separate rules in which only one weight $c$ is non-zero:
\[(\beta_1^1 \rightarrow c_1^1) (\beta_2^1 \rightarrow 0) \dots (\beta_j^i \rightarrow 0) \dots (\beta_{k_m}^m \rightarrow 0)\]
\[(\beta_1^1 \rightarrow 0) (\beta_2^1 \rightarrow c_2^1) \dots (\beta_j^i \rightarrow 0) \dots (\beta_{k_m}^m \rightarrow 0)\]
\[\dots\]
\[(\beta_1^1 \rightarrow 0) (\beta_2^1 \rightarrow 0) \dots (\beta_j^i \rightarrow 0) \dots (\beta_{k_m}^m \rightarrow c_{k_m}^m).\]
The size of this set of rules is $O(S^3)$. 

Finally, for every player $p \in N \setminus \oplus$ we add an MC-nets rule $((p) \rightarrow 0)$ to every rule. 
We added at most $n$ MC-nets rules to $O(S)$ rules; hence, the total size of the final set of rules is $O(S^3 + n S)$.
This concludes the proof.
\end{proof}

\begin{example}\label{example:hybrid_weighted}
Recall the weighted MC-nets rule from \cref{example:weighted_mcnets}:
\[ (1 \land 2 \land \lnot 4 \rightarrow c_1) \ (6 \land \lnot 4 \rightarrow c_2) \mid (3 \land 5 \land \lnot 4 \rightarrow c_3).\]
First, combine all rules into one block: we have $\oplus^1 = \{1,2,6\}$ and $\oplus^2 = \{3,5\}$, so we get:
\[ (1 \land 2 \land \lnot 3 \land \lnot 4 \land \lnot 5\rightarrow c_1) \ (6 \land \lnot 3 \land \lnot 4 \land \lnot 5 \rightarrow c_2) \ (3 \land 5 \land \lnot 1 \land \lnot 2 \land \lnot 4 \land \lnot 6 \rightarrow c_3).\]
Now, we split this rule into three and add $(4 \rightarrow 0)$ to each rule, as $4$ does not appear as a positive literal in any rule:
\[ (1 \land 2 \land \lnot 3 \land \lnot 4 \land \lnot 5\rightarrow c_1) (6 \land \lnot 3 \land \lnot 4 \land \lnot 5 \rightarrow 0) (3 \land 5 \land \lnot 1 \land \lnot 2 \land \lnot 4 \land \lnot 6 \rightarrow 0) (4 \rightarrow 0), \]
\[ (1 \land 2 \land \lnot 3 \land \lnot 4 \land \lnot 5\rightarrow 0) (6 \land \lnot 3 \land \lnot 4 \land \lnot 5 \rightarrow c_2) (3 \land 5 \land \lnot 1 \land \lnot 2 \land \lnot 4 \land \lnot 6 \rightarrow 0) (4 \rightarrow 0),\]
\[ (1 \land 2 \land \lnot 3 \land \lnot 4 \land \lnot 5\rightarrow 0) (6 \land \lnot 3 \land \lnot 4 \land \lnot 5 \rightarrow 0) (3 \land 5 \land \lnot 1 \land \lnot 2 \land \lnot 4 \land \lnot 6 \rightarrow c_3) \ (4 \rightarrow 0).\]
This concludes the construction.
\end{example}

In turn, in \cref{lemma:hybrid_rules_embedded}, we show that embedded MC-nets rules are equivalent to a subset of hybrid rules.

\begin{definition}\label{definition:simple_hybrid_rules} (Regular hybrid rules)
A hybrid rule is \emph{regular} if for every $\alpha_i,\alpha_j$, ($i,j>1$) compatible with $\alpha_1$ it holds $\alpha_i \sim \alpha_j$ and $|{\textstyle\oplus}(\alpha_i)|=1$.
%\[
%\forall_{1 < i,j \leq k} ((\alpha_1 \sim \alpha_i) \land (\alpha_1 \sim \alpha_j)) \rightarrow ((\alpha_i \sim \alpha_j) \land (|{\textstyle\oplus}(\alpha_i)|=|{\textstyle\oplus}(\alpha_j)|=1)).
%\]
\end{definition}

\begin{lemma}\label{lemma:hybrid_rules_embedded}
Every embedded MC-nets rule of size $S$ is equivalent to a regular hybrid rule of size $\mathrm{poly}(n,S)$.
Moreover, every regular hybrid rule of size $S$ is equivalent to an embedded MC-nets rule of size $\mathrm{poly}(n,S)$.
\end{lemma}

\begin{proof}
Consider an embedded MC-nets rule:
\[ (\alpha_1 | \alpha_2, \dots, \alpha_k) \rightarrow c .\]
For this rule we will construct an equivalent regular hybrid rule.

Assume $\oplus(\alpha_i) \cap \oplus(\alpha_j) \neq \emptyset$ for some $1 \leq i < j \leq k$. 
If $i=1$ or $\alpha_i$ and $\alpha_j$ are not compatible, then the embedded MC-nets rule is contradictory; hence it is equivalent to any hybrid rule with $c=0$.
If $i,j > 1$ and $\alpha_i$ and $\alpha_j$ are compatible, then $\alpha_i$ and $\alpha_j$ can be replaced by $\alpha_i \land \alpha_j$.
Hence, in what follows, we assume that $\oplus(\alpha_i) \cap \oplus(\alpha_j) = \emptyset$ for every $i,j \in \{1,\dots,k\}$.

Let us define expressions $\beta_1$ as $\alpha_1$ with $\oplus(\alpha_2) \cup \dots \cup \oplus(\alpha_k)$ added as negative literals.
Since $\beta_1$ and $\alpha_i$ for any $i > 1$ cannot be now satisfied by the same coalition, we get that the following weighted MC-nets rule is equivalent to the original one:
\[ (\beta_1 \rightarrow c) (\alpha_2 \rightarrow 0) \dots (\alpha_k \rightarrow 0). \]
Note that this may not be a hybrid rule, since not all players may appear as positive literals. 
Hence, we add an MC-nets rule $((p) \rightarrow 0)$ for every player $p \in N \setminus \bigcup_{i=1}^k \oplus(\alpha_i)$.

Note that the resulting hybrid rule is regular: only MC-nets rules of the form $((p) \rightarrow 0)$ are compatible with $\beta_1$ and indeed each of them has one positive literal and they are all compatible with each other. 
For the original embedded MC-nets rule of size $S$, the size of the resulting hybrid rule is $O(S^2 + N)$.

Now, consider a regular hybrid rule:
\[ (\alpha_1 \rightarrow c) (\alpha_2 \rightarrow 0) \dots (\alpha_k \rightarrow 0). \]
For this rule we will construct an equivalent embedded MC-nets rule.

Without loss of generality assume expressions $\alpha_{m+1}, \dots, \alpha_k$ are compatible with $\alpha_1$.
Fix $i \in \{m+1,\dots,k\}$.
We know that $|\!\oplus\!(\alpha_i)| = 1$.
Assume $\ominus(\alpha_i) \neq \emptyset$ and take player $p \in \ominus(\alpha_i)$. 
From the definition of hybrid rules we know that every player appears as a positive literal somewhere, i.e., there exists $\alpha_j$ such that $p \in \oplus(\alpha_j)$. 
Hence, we can remove $\lnot p$ from $\alpha_i$ and add a player from $\oplus(\alpha_i)$ as a negative literal in $\alpha_j$ (unless it is already in $\ominus(\alpha_j)$---in such a case addition can be omitted) without changing the satisfiability of the rule. 
From the definition of regular hybrid rules we get that $\alpha_j$ is not compatible with $\alpha_1$; hence, the rule remain regular.

By doing so for every negative literal of every rule compatible with $\alpha_1$ (other than $\alpha_1$), we will obtain a rule of the form:
\[ (\alpha_1 \rightarrow c) (\beta_2 \rightarrow 0) \dots (\beta_m \rightarrow 0) ((p_1) \rightarrow 0) \dots ((p_{k-m}) \rightarrow 0), \]
in which all $\beta_2,\dots,\beta_m$ are not compatible with $\alpha_1$. 
Hence, this rule is equivalent to an embedded MC-nets rule:
\[ (\alpha_1 | \beta_2, \dots, \beta_m) \rightarrow c. \]
Note that the size of the resulting embedded MC-nets rule is smaller than the size of the original hybrid rule. This concludes the proof.
\end{proof}

\begin{example}\label{example:hybrid_embedded}
Recall the embedded MC-nets rule from \cref{example:embedded_mcnets}:
\[ \gamma = (1 \land 2 \rightarrow 1) \mid (3 \land 5 \land \lnot 4 \land \lnot 6) \ (4 \land \lnot 6) \]
To obtain a hybrid rule we add negative literals to the first MC-nets rule, add $(6 \rightarrow 0)$ and get:
\[(1 \land 2 \land \lnot 3 \land \lnot 4 \land \lnot 5 \rightarrow 1) \ (3 \land 5 \land \lnot 4 \land \lnot 6 \rightarrow 0) \ (4 \land \lnot 6 \rightarrow 0) \ (6 \rightarrow 0) \]
Only the last MC-nets rule is compatible with the first rule and it contains one positive literal. 
Hence, the rule is regular.

Note that this rule is equivalent to the first hybrid rule obtained in \cref{example:hybrid_weighted}.
\end{example}

\subsection{Player-graphs}\label{section:player-graphs}
So far, we have shown the mapping from weighted MC-nets and embedded MC-nets rules to hybrid rules.
In what follows, we show that every hybrid rule can be represented as a graph and the corresponding game can be defined based on colorings in such a graph.
We will consider undirected graphs with nodes labeled by a partition of players and one node highlighted.
We will call them \emph{player-graphs}.

\begin{definition} (Player-graphs)
A \emph{player-graph} is a tuple $G = (V,E, l, v^*)$, where $V$ is the set of nodes, $E$ is the set of undirected edges, i.e., subsets of nodes of size 2, $l: V \rightarrow 2^N$ is a node label function such that node labels form a partition of $N$ (formally: $l(v) \cap l(u) = \emptyset$ for every $u,v \in V$ and $\bigcup_{v \in V} l(v) = N$) and $v^* \in V$ is a highlighted node. 
\end{definition}

Let us introduce some basic graph definitions.
Node $u$ is \emph{adjacent} to $v$ if $\{u, v\} \in E$.
The set of nodes adjacent to node $v$, called \emph{neighbors}, is denoted by $\mathcal{N}(v)$.
A \emph{clique} is a subset of nodes every two of which are adjacent.
An \emph{independent set} is a subset of nodes no two of which are adjacent.

A (proper vertex) \emph{$k$-coloring} of a graph is a function, $f: V \rightarrow \{1,\dots,k\}$, that assigns colors $\{1,\dots,k\}$ to nodes in a way that every two adjacent nodes have different colors, i.e., $f(v) \neq f(u)$ for every $\{v,u\} \in E$. 
In other words, nodes colored with the same color form an independent set.
The set of all $k$-colorings of a graph $G$ is denoted by $C_k(G)$.

We will need some additional notation regarding colorings in players-graphs.
A $k$-coloring $f$ results in the partition of nodes:
\[ VP_f = \{f^{-1}(i) : i \in \{1,\dots,k\}, f^{-1}(i) \neq \emptyset\}. \]
The set in $VP_f$ that contains highlighted node $v^*$ is denoted by $VS_f^*$.
The partition of nodes imposes a partition of players in the player-graph is denoted by $P_f$:
\[ P_f = \left\{ \bigcup_{v \in U} l(v) : U \in VP_f \right\}. \]
Clearly, $|P_f| = |VP_f|$.
The set in $P_f$ that correspond to $VS^*_f$, i.e., that contains players from the label of $v^*$, will be denoted by $S_f^*$.

We say that $f$ uses exactly $p$ colors if $|VP_f|=p$. 
We will denote by $\#_f$ the number of all $k$-colorings that result in the partition $VP_f$; note that $\#_f = k(k-1) \cdots (k-|VP_f|+1)$.

\begin{figure}[t]
\centering
\begin{tikzpicture}[x=0.25\textwidth,y=0.25\textwidth]
  \node[main node, color1,minimum size=1cm] (1) at (0.00, 2.80) {};
  
  \node[main node, label={left:$v_1$}, color1] (1) at (0.00, 2.80) {$1$,$2$};
  \node[main node, label={left:$v_2$}, color3] (2) at (0.50, 3.05) {$3$,$5$};
  \node[main node, label={left:$v_3$}, color2] (3) at (0.50, 2.55) {$4$};
  \node[main node, label={left:$v_4$}, color1] (4) at (1.00, 2.80) {$6$};

  \path[draw,thick]
  (1) edge (2)
  (1) edge (3)
  (2) edge (3)
  (3) edge (4)
  (2) edge (4)
  ;

\end{tikzpicture}
\caption{Player-graph $(V,E,l,v_1)$ and a 4-coloring $f$ with colors: 1~(blue/striped), 2~(yellow/checked), 3~(red/plain), 4~(green/dotted). Note that color 2 is not used.} 
\label{figure:g_gamma}
\end{figure}

Let us now define a player-graph that will represent a hybrid rule.

\begin{definition}\label{definition:g_gamma} (Graph $G^{\gamma}$)
For a hybrid rule $\gamma$, player-graph $G^{\gamma} = (V, E, l, v^*)$ is a graph where nodes represent expressions $\alpha_1,\dots,\alpha_k$ and are labeled with sets $\oplus(\alpha_1), \dots, \oplus(\alpha_k)$ and edges connect incompatible expressions $\alpha_i, \alpha_j$:
\begin{itemize}
\item $V = \{v_1,\dots,v_k\}$, $v^* = v_1$ and $l(v_i) = \oplus(\alpha_i)$ for every $v_i \in V$; and
\item $E = \{\{v_i,v_j\} \subseteq V : \alpha_i \not\sim \alpha_j\}$.
\end{itemize}
\end{definition}
\noindent We note that a similar construction of a graph for the right-hand side part of the embedded MC-nets was proposed in \cite{Skibski:etal:2016:kcoalitional}.

\begin{example}\label{example:g_gamma}
Recall a hybrid rule from \cref{example:hybrid_embedded}:
$(\alpha_1 \rightarrow 1) (\alpha_2 \rightarrow 0) (\alpha_3 \rightarrow 0) (\alpha_4 \rightarrow 0)$
with
\[ \alpha_1 = (1 \land 2 \land \lnot 3 \land \lnot 4 \land \lnot 5), \quad \alpha_2 = (3 \land 5 \land \lnot 4 \land \lnot 6), \quad \alpha_3 = (4 \land \lnot 6), \quad \alpha_4 = (6) \]
Note that $\alpha_1 \sim \alpha_4$ and $\alpha_i \not\sim \alpha_j$ for other $i,j \in \{1,\dots,4\}$, $i \neq j$.
Hence, we obtain a graph depicted in Figure~\ref{figure:g_gamma}.

The presented coloring $f$ induces the following partition of nodes and, consequently, of players:
\[ VP_f = \{\{v_1, v_4\}, \{v_2\}, \{v_3\}\}, \quad P_f = \{\{1,2,6\}, \{3,5\}, \{4\}\}. \]
There are $\#_f = 4 \cdot 3 \cdot 2 = 24$ different 4-colorings that result in the same partitions.
We have $VS_f^* = \{v_1, v_4\}$ and $S_f^* = \{1,2,6\}$.
\end{example}

Definition~\ref{definition:g_gamma} shows that every hybrid rule can be represented as a player-graph.
Also, every player-graph represents some hybrid rule, which we prove in the following lemma.

\begin{lemma}\label{lemma:g_gamma_weighted}
For every player-graph $G = (V,E,l,v^*)$ there exists a hybrid rule $\gamma$ s.t. $G = G^{\gamma}$.
\end{lemma}

\begin{proof}
Let $G = (V,E,l,v^*)$ be a player-graph and without loss of generality assume $V = \{v_1,\dots,v_k\}$ and $v_1 = v^*$.
We can construct a hybrid rule $\gamma$ as follows: for every node $v_i \in V$ we create a boolean expression $\alpha_i$ as in \cref{eq:mcnet_form_alpha} with $l(v_i)$ as positive literals and $\bigcup_{v_j \in \mathcal{N}(v_i)} l(v_j)$ as negative literals. 
Now, for $\gamma = (\alpha_1 \rightarrow 1) (\alpha_2 \rightarrow 0) \dots (\alpha_k \rightarrow0)$ we have $G = G^{\gamma}$.
\end{proof}

The next lemma states the necessary and sufficient conditions for the graph to represent a hybrid rule that is regular.
We will call such player-graphs \emph{regular} as well.

\begin{definition} (Regular player-graphs)
A player-graph $G = (V,E,l,v^*)$ is \emph{regular} if nodes from $V \setminus (\mathcal{N}(v^*) \cup \{v^*\})$ form an independent set and have singleton labels, i.e., $|l(u)|=1$ for every $u \in V \setminus (\mathcal{N}(v^*) \cup \{v^*\})$.
\end{definition}

\begin{lemma}\label{lemma:g_gamma_embedded}
For a player-graph $G = (V,E,l,v^*)$ there exists a regular hybrid rule $\gamma$ such that $G = G^{\gamma}$ if and only if the player-graph is regular.
\end{lemma}

\begin{figure}[t]
\centering
\begin{tikzpicture}[x=5cm,y=5cm] % change these values to adjust the size of a figure
  \node[main node,minimum size=1cm] (1) at (-0.10, 1.18) {};
  \node[main node,label={left:$v^*$}] (1) at (-0.10, 1.18) {$1$}; 
  \node[main node] (2) at (0.42, 1.35) {$2$}; 
  \node[main node] (3) at (0.62, 1.12) {$5$,$7$}; 
%  \node[main node] (4) at (0.38, 1.07) {$10$}; 
  \node[main node] (5) at (0.43, 0.86) {$3$,$6$}; 
  \node[main node] (6) at (1.10, 1.31) {$4$}; 
  \node[main node] (7) at (1.10, 1.09) {$8$}; 
  \node[main node] (8) at (1.10, 0.86) {$9$}; 

  \draw[dashed] (0.22,0.72) rectangle (0.75, 1.49);
  \node at (0.1, 0.78) {$\mathcal{N}(v^*)$};

  \path[draw,thick]
  (3) edge  (5)
  (5) edge  (1)
  (3) edge  (2)
  (1) edge  (3)
  (2) edge  (1)
%  (1) edge  (4)
%  (2) edge  (4)
%  (4) edge  (5)
%  (4) edge  (3)
  (6) edge  (3)
  (7) edge  (5)
  (7) edge  (3)
  ;
\end{tikzpicture}
\caption{The structure of regular player-graphs, i.e., player-graphs that represent embedded MC-nets rule: nodes not connected to $v^*$ form an independent set and have singleton labels.}
\label{figure:graph_gamma_embedded}
\end{figure}

\begin{proof}
Fix a regular hybrid rule $\gamma$.
From the definition we know that all expressions compatible with $\alpha_1$ are compatible with each other and have a single positive literal.
Hence, nodes in $G^{\gamma}$ not adjacent to $v^*$ form an independent set and have labels of size $1$.

On the other hand, let $G = (V,E,l,v^*)$ be a regular player-graph and without loss of generality assume $V = \{v_1,\dots,v_k\}$ and $v_1 = v^*$.
To construct a regular hybrid rule $\gamma$ we repeat the construction from the proof of \cref{lemma:g_gamma_weighted}: for every node $v_i \in V$ we create a boolean expression $\alpha_i$ as in \cref{eq:mcnet_form_alpha} with $l(v_i)$ as positive literals and $\bigcup_{v_j \in \mathcal{N}(v_i)} l(v_j)$ as negative literals. 
Now, it is easy to verify that for $\gamma = (\alpha_1 \rightarrow 1) (\alpha_2 \rightarrow 0) \dots (\alpha_k \rightarrow 0)$ we have $G = G^{\gamma}$.
Clearly, if $v_i$ and $v_j$ are not adjacent, then $\alpha_i$ and $\alpha_j$ are compatible. 
Hence, from the definition of regular player-graphs we get that $\gamma$ is regular.
\end{proof}

Based on Lemmas~\ref{lemma:hybrid_rules_weighted}--\ref{lemma:g_gamma_embedded} we know that every weighted MC-nets rule can be represented as (one or more) player-graphs and every embedded MC-nets rule can be represented as a regular player-graph.

Let us now explain how a game represented as a hybrid rule $\gamma$ can be defined based on graph $G^{\gamma}$.
Fix a hybrid rule $\gamma$ and consider a partition $P$ that satisfies it. 
Since every node in graph $G^{\gamma}$ is labeled with a set of players which is equal to the set of positive literals in some expression $\alpha_i$, it is clear that all these players must appear in the same coalition in $P$.
This observation combined with the fact that every player appears in exactly one node implies that $P$ can be associated with a partition of nodes in graph $G^{\gamma}$.

Consider partition $VP$ of nodes that correspond to $P$.
Note that two adjacent nodes cannot belong to the same coalition in $VP$, because they represent incompatible expressions that cannot be satisfied by one coalition.
Hence, every set in $VP$ is an independent set. 
As a result, we get that $VP$ corresponds to some coloring of a graph.

This analysis is formalized in the following lemma.
Recall that $k$ is the number of boolean expressions in $\gamma$ and nodes in $G^{\gamma}$.

\begin{lemma}\label{lemma:game_from_g_gamma}
Partition $P$ satisfies a hybrid rule $\gamma$ if and only if there exists a $k$-coloring $f \in C_k(G^{\gamma})$ such that $P = P_f$.
Moreover, if the weight of $\gamma$ is $c$, then for every $(S,P) \in EC$:
\begin{equation}\label{eq:lemma:game_from_g_gamma}
g^{\gamma}(S,P) = \sum_{f \in C_k(G^{\gamma}) : (S_f^*,P_f) = (S,P)} \frac{c}{\#_f}.
\end{equation}
\end{lemma}
%\[ g^{\gamma}(S,P) = [\exists_{f \in C_k(G^{\gamma})} (P_f = P \land S_f^* = S)].\]
\begin{proof}
Assume $P = \{S_1, \dots, S_m\}$ satisfies $\gamma$.
We have that:
\begin{itemize}
\item[(i)] every expression $\alpha_i$ is satisfied by exactly one coalition (it is not possible that $\oplus(\alpha_i)$ is a subset of two non-overlapping coalitions); 
\item[(ii)] every coalition $S_r$ satisfies at least one expression (for an arbitrary player $p \in S_r$ we know that there exists an expression $\alpha_j$ such that $p \in \oplus(\alpha_j)$; hence only $S_r$ may satisfy $\alpha_j$).
\end{itemize}
Let us define a function $f$ as follows: $f(v_i) = r$ such that $\oplus(\alpha_i) \subseteq S_r$ (from (i) we know there exists exactly one such $r$).
Function $f$ is a proper coloring: if $f(v_i) = f(v_j) = r$, then coalition $S_r$ satisfies both $\alpha_i$ and $\alpha_j$, hence they are compatible and $\{v_i,v_j\} \not \in E$. 
Also, $f$ is a $k$-coloring, because $m \le k$ (from (i) and (ii)).

Now, take $k$-coloring $f$ of $G^{\gamma}$ and consider $P_f$.
Fix $v_i \in V$ and let $S \in P_f$ be the set of all players in nodes colored with the same color as node $v_i$.
We claim $S$ satisfies $\alpha_i$. 
Obviously, $\oplus(\alpha_i) \subseteq S$. 
On the other hand, $\ominus(\alpha_i) \cap S = \emptyset$, because players from $\ominus(\alpha_i)$ all appear in nodes which are neighbors of $v_i$ in graph $G^{\gamma}$, hence have different colors than $v_i$.

So far, we have proved that $P$ satisfies a hybrid rule $\gamma$ if and only if there exists a $k$-coloring $f \in C_k(G^{\gamma})$ such that $P = P_f$. 
Now, $(S,P)$ has non-zero value if and only if $P$ satisfies $\gamma$ and $S$ satisfies $\alpha_1$ which---in a partition satisfying $\gamma$---is equivalent to containing all players from nodes colored with the same color as node $v^*$. 
If $P = P_f$, then there are $\#_f$ other $k$-colorings that results in the same partition.
This proves \cref{eq:lemma:game_from_g_gamma}.
\end{proof}

\begin{example}\label{example:game_from_g_gamma}
Consider a hybrid rule $\gamma$ from \cref{example:g_gamma} with graph $G^{\gamma}$ depicted in Figure~\ref{figure:g_gamma}.
Let us discuss all possible 4-colorings of $G^{\gamma}$:
\begin{itemize}[label={--}]
\item There are no colorings of $G^{\gamma}$ that use 1 or 2 colors.
\item There are $24$ colorings that use 3 colors: $f(v_1)=f(v_4)=a$, $f(v_2)=b$ and $f(v_3)=c$ where $a,b,c \in \{1,\dots,4\}$ are different colors.
Note that for every such coloring $f$ we have $VP_f = \{\{v_1,v_4\}, \{v_2\}, \{v_3\}\}$ and $P_f = \{\{1,2,6\}, \{3,5\}, \{4\}\}$.
\item There are $24$ colorings that use 4 colors; in these colorings all nodes have different colors. For every such coloring $f$ we get $P_f = \{\{1,2\}, \{3,5\}, \{4\}, \{6\}\}$
\end{itemize}
Overall, 48 colorings results in two partitions of players.
Now, from \cref{lemma:game_from_g_gamma}, game $g^{\gamma}$ is defined as follows:
\[ g^{\gamma}(S,P) = \begin{cases}
1 & \mbox{ if }(S,P)= (\{1,2,6\}, \{\{1,2,6\}, \{3,5\}, \{4\}\}) \\
 & \mbox{ or } (S,P) = (\{1,2\}, \{\{1,2\}, \{3,5\}, \{4\}, \{6\}\}), \\
0 & \mbox{ otherwise.}
\end{cases} \]
%\[ g^{\gamma}(\{1,2,6\}, \{\{1,2,6\}, \{3,5\}, \{4\}\}) = 1, \quad g^{\gamma}(\{1,2\}, \{\{1,2\}, \{3,5\}, \{4\}, \{6\}\}) = 1. \]
This agrees with \cref{example:embedded_mcnets}.
\end{example}

\section{Computing extended Shapley values}\label{section:computing}
In this section, building upon our analysis from the previous section, we consider computing extended Shapley values in games represented as embedded and weighted MC-nets.

All extended Shapley values considered by us satisfy linearity, i.e., $ESV(g + g') = ESV(g) + ESV(g')$ and $ESV(c \cdot g) = c \cdot ESV(g)$  for every two games $g, g'$ and $c \in \mathbb{R}$.
Thus, in our computational analysis we can focus on games represented as a single rule and, based on Lemmas~\ref{lemma:hybrid_rules_weighted} and \ref{lemma:hybrid_rules_embedded}, as a single hybrid rule. 
Moreover, we can assume the weight of this rule is $1$ (i.e., $c = 1$).
Hence, from now on, we will assume that game is represented as a hybrid rule with weight $1$.

Fix such a hybrid rule $\gamma$.
From \cref{eq:extended_sv_general_form} and \cref{lemma:game_from_g_gamma} we get the following formula for extended Shapley values:
\begin{equation}\label{eq:esv_is_colorings}
ESV_i(g^{\gamma}) = \sum_{(S,P) \in EC} \left(\omega_i(S,P) \sum_{f \in C_k(G^{\gamma}) : (S_f^*,P_f)=(S,P)} \frac{1}{\#_f} \right) = \sum_{f \in C_k(G^{\gamma})} \frac{\omega_i(S_f^*,P_f)}{\#_f}.
\end{equation}
To put it in words, extended Shapley value in game $g^{\gamma}$ is a weighted sum over all colorings in graph $G^{\gamma}$.
Weights depend on $P_f$ (partition of players resulting from the coloring $f$), $S_f^*$ (set of players from labels of nodes colored with the same color as node $v^*$), and player $i \in N$.
For an extensive example, see Table~\ref{table:big_example}.

\begin{table*}[p]
\def\arraystretch{1.1}
\setcellgapes{2pt}
\makegapedcells
\setlength{\tabcolsep}{4.5pt}
\centering
\scalebox{0.67}{
\begin{tabular}{|c|c|c|c|c|c|c|c|c|c|}
\hhline{|=|=|=|=|=|=====|}
\multirow{2}{*}{$f$} & \multirow{2}{*}{$|P_f|$} & \multirow{2}{*}{$\#_f$} & \multirow{2}{*}{$S_f^*$} & \multirow{2}{*}{$P_f \setminus \{S_f^*\}$} & \multicolumn{5}{c|}{$\omega_1(S_f^*, P_f)$}\\\cline{6-10}
 & &  &  &  & MQ & EF & HY & SS & MY\\
\hhline{|=|=|=|=|=|=|=|=|=|=|}
\begin{tikzpicture}[x=1.2cm,baseline=(current bounding box.center)]
  \node[main node, color1, minimum size=1.01cm] (1) at (0,0) {};
  \node[main node, color1] (1) at (0,0) {$1$};
  \node[main node, color2] (2) at (1,0) {$2$};
  \node[main node, color1] (3) at (2,0) {$3$};
  \node[main node, color1] (4) at (3,0) {$4$,$6$};
  \node[main node, color2] (5) at (4,0) {$5$};

  \path[draw,thick] (1) edge (2) (2) edge (3) (4) edge (5);
\end{tikzpicture} 
& 2 & 20 & $\{1,\!3,\!4,\!6\}$ & $\{\{2,5\}\}$ & $1$ & 0 & $52$ & 6 & 5\\
\hline
\begin{tikzpicture}[x=1.2cm,baseline=(current bounding box.center)]
  \node[main node, color1, minimum size=1.01cm] (1) at (0,0) {};
  \node[main node, color1] (1) at (0,0) {$1$};
  \node[main node, color2] (2) at (1,0) {$2$};
  \node[main node, color1] (3) at (2,0) {$3$};
  \node[main node, color2] (4) at (3,0) {$4$,$6$};
  \node[main node, color1] (5) at (4,0) {$5$};

  \path[draw,thick] (1) edge (2) (2) edge (3) (4) edge (5);
\end{tikzpicture} 
& 2 & 20 & $\{1,3,5\}$ & $\{\{2,4,6\}\}$ & $1$ & 0 & $15$ & 4 & 10\\
\hhline{|=|=|=|=|=|=|=|=|=|=|}
\begin{tikzpicture}[x=1.2cm,baseline=(current bounding box.center)]
  \node[main node, color1, minimum size=1.01cm] (1) at (0,0) {};
  \node[main node, color1] (1) at (0,0) {$1$};
  \node[main node, color2] (2) at (1,0) {$2$};
  \node[main node, color1] (3) at (2,0) {$3$};
  \node[main node, color1] (4) at (3,0) {$4$,$6$};
  \node[main node, color3] (5) at (4,0) {$5$};

  \path[draw,thick] (1) edge (2) (2) edge (3) (4) edge (5);
\end{tikzpicture} 
& 3 & 60 & $\{1,\!3,\!4,\!6\}$ & $\{\{2\}, \{5\}\}$ & 0 & $1$ & $151$ & $6$ & -4\\
\hline
\begin{tikzpicture}[x=1.2cm,baseline=(current bounding box.center)]
  \node[main node, color1, minimum size=1.01cm] (1) at (0,0) {};
  \node[main node, color1] (1) at (0,0) {$1$};
  \node[main node, color2] (2) at (1,0) {$2$};
  \node[main node, color1] (3) at (2,0) {$3$};
  \node[main node, color3] (4) at (3,0) {$4$,$6$};
  \node[main node, color1] (5) at (4,0) {$5$};

  \path[draw,thick] (1) edge (2) (2) edge (3) (4) edge (5);
\end{tikzpicture} 
& 3 & 60 & $\{1,3,5\}$ & $\{\{2\}, \{4,6\}\}$ & 0 & 0 & $37$ & $2$ & -7\\
\hline
\begin{tikzpicture}[x=1.2cm,baseline=(current bounding box.center)]
  \node[main node, color1, minimum size=1.01cm] (1) at (0,0) {};
  \node[main node, color1] (1) at (0,0) {$1$};
  \node[main node, color2] (2) at (1,0) {$2$};
  \node[main node, color3] (3) at (2,0) {$3$};
  \node[main node, color1] (4) at (3,0) {$4$,$6$};
  \node[main node, color2] (5) at (4,0) {$5$};

  \path[draw,thick] (1) edge (2) (2) edge (3) (4) edge (5);
\end{tikzpicture} 
& 3 & 60 & $\{1,4,6\}$ & $\{\{2,5\}, \{3\}\}$ & 0 & 0 & $37$ & $2$ & -7\\
\hline
\begin{tikzpicture}[x=1.2cm,baseline=(current bounding box.center)]
  \node[main node, color1, minimum size=1.01cm] (1) at (0,0) {};
  \node[main node, color1] (1) at (0,0) {$1$};
  \node[main node, color2] (2) at (1,0) {$2$};
  \node[main node, color3] (3) at (2,0) {$3$};
  \node[main node, color1] (4) at (3,0) {$4$,$6$};
  \node[main node, color3] (5) at (4,0) {$5$};

  \path[draw,thick] (1) edge (2) (2) edge (3) (4) edge (5);
\end{tikzpicture} 
& 3 & 60 & $\{1,4,6\}$ & $\{\{2\}, \{3,5\}\}$ & 0 & 0 & $37$ & $2$ & -7\\
\hline
\begin{tikzpicture}[x=1.2cm,baseline=(current bounding box.center)]
  \node[main node, color1, minimum size=1.01cm] (1) at (0,0) {};
  \node[main node, color1] (1) at (0,0) {$1$};
  \node[main node, color2] (2) at (1,0) {$2$};
  \node[main node, color1] (3) at (2,0) {$3$};
  \node[main node, color3] (4) at (3,0) {$4$,$6$};
  \node[main node, color2] (5) at (4,0) {$5$};

  \path[draw,thick] (1) edge (2) (2) edge (3) (4) edge (5);
\end{tikzpicture} 
& 3 & 60 & $\{1,3\}$ & $\{\{2,5\}, \{4,6\}\}$ & 0 & 0 & $20$ & $1$ & -10\\
\hline
\begin{tikzpicture}[x=1.2cm,baseline=(current bounding box.center)]
  \node[main node, color1, minimum size=1.01cm] (1) at (0,0) {};
  \node[main node, color1] (1) at (0,0) {$1$};
  \node[main node, color2] (2) at (1,0) {$2$};
  \node[main node, color1] (3) at (2,0) {$3$};
  \node[main node, color2] (4) at (3,0) {$4$,$6$};
  \node[main node, color3] (5) at (4,0) {$5$};

  \path[draw,thick] (1) edge (2) (2) edge (3) (4) edge (5);
\end{tikzpicture} 
& 3 & 60 & $\{1,3\}$ & $\{\{2,4,6\}, \{5\}\}$ & 0 & 0 & $20$ & $2$ & -12\\
\hline
\begin{tikzpicture}[x=1.2cm,baseline=(current bounding box.center)]
  \node[main node, color1, minimum size=1.01cm] (1) at (0,0) {};
  \node[main node, color1] (1) at (0,0) {$1$};
  \node[main node, color2] (2) at (1,0) {$2$};
  \node[main node, color3] (3) at (2,0) {$3$};
  \node[main node, color2] (4) at (3,0) {$4$,$6$};
  \node[main node, color1] (5) at (4,0) {$5$};

  \path[draw,thick] (1) edge (2) (2) edge (3) (4) edge (5);
\end{tikzpicture} 
& 3 & 60 & $\{1,5\}$ & $\{\{2,4,6\}, \{3\}\}$ & 0 & 0 & $20$ & $2$ &  -12\\
\hline
\begin{tikzpicture}[x=1.2cm,baseline=(current bounding box.center)]
  \node[main node, color1, minimum size=1.01cm] (1) at (0,0) {};
  \node[main node, color1] (1) at (0,0) {$1$};
  \node[main node, color2] (2) at (1,0) {$2$};
  \node[main node, color3] (3) at (2,0) {$3$};
  \node[main node, color3] (4) at (3,0) {$4$,$6$};
  \node[main node, color1] (5) at (4,0) {$5$};

  \path[draw,thick] (1) edge (2) (2) edge (3) (4) edge (5);
\end{tikzpicture} 
& 3 & 60 & $\{1,5\}$ & $\{\{2\}, \{3,4,6\}\}$ & 0 & 0 & $20$ & $2$ & -12\\
\hline
\begin{tikzpicture}[x=1.2cm,baseline=(current bounding box.center)]
  \node[main node, color1, minimum size=1.02cm] (1) at (0,0) {};
  \node[main node, color1] (1) at (0,0) {$1$};
  \node[main node, color2] (2) at (1,0) {$2$};
  \node[main node, color3] (3) at (2,0) {$3$};
  \node[main node, color2] (4) at (3,0) {$4$,$6$};
  \node[main node, color3] (5) at (4,0) {$5$};

  \path[draw,thick] (1) edge (2) (2) edge (3) (4) edge (5);
\end{tikzpicture} 
& 3 & 60 & $\{1\}$ & $\{\{2,4,6\}, \{3,5\}\}$ & 0 & 0 & $30$ & $2$ & -15\\
\hline
\begin{tikzpicture}[x=1.2cm,baseline=(current bounding box.center)]
  \node[main node, color1, minimum size=1.01cm] (1) at (0,0) {};
  \node[main node, color1] (1) at (0,0) {$1$};
  \node[main node, color2] (2) at (1,0) {$2$};
  \node[main node, color3] (3) at (2,0) {$3$};
  \node[main node, color3] (4) at (3,0) {$4$,$6$};
  \node[main node, color2] (5) at (4,0) {$5$};

  \path[draw,thick] (1) edge (2) (2) edge (3) (4) edge (5);
\end{tikzpicture} 
& 3 & 60 & $\{1\}$ & $\{\{2,5\}, \{3,4,6\}\}$ & 0 & 0 & $30$ & $2$ & -15\\
\hhline{|=|=|=|=|=|=|=|=|=|=|}
\begin{tikzpicture}[x=1.2cm,baseline=(current bounding box.center)]
  \node[main node, color1, minimum size=1.01cm] (1) at (0,0) {};
  \node[main node, color1] (1) at (0,0) {$1$};
  \node[main node, color2] (2) at (1,0) {$2$};
  \node[main node, color3] (3) at (2,0) {$3$};
  \node[main node, color1] (4) at (3,0) {$4$,$6$};
  \node[main node, color4] (5) at (4,0) {$5$};

  \path[draw,thick] (1) edge (2) (2) edge (3) (4) edge (5);
\end{tikzpicture} 
& 4 & 120 & $\{1,4,6\}$ & $\{\{2\}, \{3\}, \{5\}\}$ & 0 & $1$ & $77$ & $2$ & 12\\
\hline
\begin{tikzpicture}[x=1.2cm,baseline=(current bounding box.center)]
  \node[main node, color1, minimum size=1.01cm] (1) at (0,0) {};
  \node[main node, color1] (1) at (0,0) {$1$};
  \node[main node, color2] (2) at (1,0) {$2$};
  \node[main node, color1] (3) at (2,0) {$3$};
  \node[main node, color3] (4) at (3,0) {$4$,$6$};
  \node[main node, color4] (5) at (4,0) {$5$};

  \path[draw,thick] (1) edge (2) (2) edge (3) (4) edge (5);
\end{tikzpicture} 
& 4 & 120 & $\{1,3\}$ & $\{\{2\}, \{4,6\}, \{5\}\}$ & 0 & 0 & $34$ & $1$ & 18\\
\hline
\begin{tikzpicture}[x=1.2cm,baseline=(current bounding box.center)]
  \node[main node, color1, minimum size=1.01cm] (1) at (0,0) {};
  \node[main node, color1] (1) at (0,0) {$1$};
  \node[main node, color2] (2) at (1,0) {$2$};
  \node[main node, color3] (3) at (2,0) {$3$};
  \node[main node, color4] (4) at (3,0) {$4$,$6$};
  \node[main node, color1] (5) at (4,0) {$5$};

  \path[draw,thick] (1) edge (2) (2) edge (3) (4) edge (5);
\end{tikzpicture} 
& 4 & 120 & $\{1,5\}$ & $\{\{2\}, \{3\}, \{4,6\}\}$ & 0 & 0 & $34$ & $1$ & 18\\
\hline
\begin{tikzpicture}[x=1.2cm,baseline=(current bounding box.center)]
  \node[main node, color1, minimum size=1.01cm] (1) at (0,0) {};
  \node[main node, color1] (1) at (0,0) {$1$};
  \node[main node, color2] (2) at (1,0) {$2$};
  \node[main node, color3] (3) at (2,0) {$3$};
  \node[main node, color4] (4) at (3,0) {$4$,$6$};
  \node[main node, color2] (5) at (4,0) {$5$};

  \path[draw,thick] (1) edge (2) (2) edge (3) (4) edge (5);
\end{tikzpicture} 
& 4 & 120 & $\{1\}$ & $\{\{2,5\}, \{3\}, \{4,6\}\}$ & 0 & 0 & $40$ & $1$ & 24\\
\hline
\begin{tikzpicture}[x=1.2cm,baseline=(current bounding box.center)]
  \node[main node, color1, minimum size=1.01cm] (1) at (0,0) {};
  \node[main node, color1] (1) at (0,0) {$1$};
  \node[main node, color2] (2) at (1,0) {$2$};
  \node[main node, color3] (3) at (2,0) {$3$};
  \node[main node, color4] (4) at (3,0) {$4$,$6$};
  \node[main node, color3] (5) at (4,0) {$5$};

  \path[draw,thick] (1) edge (2) (2) edge (3) (4) edge (5);
\end{tikzpicture} 
& 4 & 120 & $\{1\}$ & $\{\{2\}, \{3,5\}, \{4,6\}\}$ & 0 & 0 & $40$ & $1$ & 24\\
\hline
\begin{tikzpicture}[x=1.2cm,baseline=(current bounding box.center)]
  \node[main node, color1, minimum size=1.01cm] (1) at (0,0) {};
  \node[main node, color1] (1) at (0,0) {$1$};
  \node[main node, color2] (2) at (1,0) {$2$};
  \node[main node, color3] (3) at (2,0) {$3$};
  \node[main node, color2] (4) at (3,0) {$4$,$6$};
  \node[main node, color4] (5) at (4,0) {$5$};

  \path[draw,thick] (1) edge (2) (2) edge (3) (4) edge (5);
\end{tikzpicture} 
& 4 & 120 & $\{1\}$ & $\{\{2,4,6\}, \{3\}, \{5\}\}$ & 0 & 0 & $40$ & $2$ & 28\\
\hline
\begin{tikzpicture}[x=1.2cm,baseline=(current bounding box.center)]
  \node[main node, color1, minimum size=1.01cm] (1) at (0,0) {};
  \node[main node, color1] (1) at (0,0) {$1$};
  \node[main node, color2] (2) at (1,0) {$2$};
  \node[main node, color3] (3) at (2,0) {$3$};
  \node[main node, color3] (4) at (3,0) {$4$,$6$};
  \node[main node, color4] (5) at (4,0) {$5$};

  \path[draw,thick] (1) edge (2) (2) edge (3) (4) edge (5);
\end{tikzpicture} 
& 4 & 120 & $\{1\}$ & $\{\{2\}, \{3,4,6\}, \{5\}\}$ & 0 & 0 & $40$ & $2$ & 28\\
\hhline{|=|=|=|=|=|=|=|=|=|=|}
\begin{tikzpicture}[x=1.2cm,baseline=(current bounding box.center)]
  \node[main node, color1, minimum size=1.01cm] (1) at (0,0) {};
  \node[main node, color1] (1) at (0,0) {$1$};
  \node[main node, color2] (2) at (1,0) {$2$};
  \node[main node, color3] (3) at (2,0) {$3$};
  \node[main node, color4] (4) at (3,0) {$4$,$6$};
  \node[main node, color5] (5) at (4,0) {$5$};

  \path[draw,thick] (1) edge (2) (2) edge (3) (4) edge (5);
\end{tikzpicture} 
& 5 & 120 & $\{1\}$ & $\{\{2\},\!\{3\},\!\{4,\!6\},\!\{5\}\}$ & 0 & 0 & $40$ & $1$ & -66\\
\hhline{|=|=|=|=|=|=|=|=|=|=|}
 & & & & & $\times\!\frac{1}{60}$ & $\times\!\frac{1}{60}$ & $\!\!\times\!\frac{1}{12180}\!\!$ & $\!\!\times\!\frac{1}{7200}\!\!$ & $\times\!\frac{1}{60}$\\
\hhline{|=|=|=|=|=|=|=|=|=|=|}
\end{tabular}
}
\caption{Example of weights associated with all possible colorings in a simple graph according to extended Shapley values. 
Graph $G = (V,E)$ has $V = \{v_1, \dots, v_5\}$ and $E = \{\{v_1, v_2\}, \{v_2, v_3\}, \{v_4, v_5\}\}$ and labels: $l(v_i) = \{i\}$ for $i \in \{1,2,3,5\}$ and $l(v_4) = \{4,6\}$. 
We consider all $5$-colorings $f: V \rightarrow \{1,\dots,5\}$ (isomorphic colorings are grouped). 
For each value, the last row contains a common multiplier.}
\label{table:big_example}
\end{table*}

More generally, we can consider the following counting problem that we name \textsc{Weighted Coloring Counting}.
The problem is parametrized with weights $\tilde{\omega}: C_{|V|}(G) \rightarrow \mathbb{R}$ that for each coloring assigns some real value.
\begin{definition}\label{definition:weighted_coloring} 
\textsc{$\tilde{\omega}$-Weighted Coloring Counting}\\
\emph{Input:} player-graph $G = (V,E,l,v^*)$, i.e., undirected graph $(V,E)$, label function $l: V \rightarrow 2^N$ s.t. labels of nodes form a partition of $N$ and $v^* \in V$\\
\emph{Output:} $\sum_{f \in C_{|V|}(G)} \tilde{\omega}(f)$.
\end{definition}

\noindent In general, this problem is computationally challenging, as it generalizes the problem of counting all $k$-colorings which is \#P-hard and allows us to determine whether a graph is 3-colorable which is NP-complete.
%However, for some weights, the problem can be done in polynomial time.

Based on \cref{eq:esv_is_colorings}, computing each extended Shapley value for a fixed player can be considered a special case of \textsc{Weighted Coloring Counting}.
In the following sections, we analyze these problems one by one.
Extended Shapley values are ordered in ascending order by the complexity of their formula:
\begin{itemize}
\item First two values, the MQ-value and the EF-value, take into account only one partition $P$ for every coalition $S$; hence, they can be computed by traversing all subsets, not all partitions of players. 
\item In the HY-value, considered third, the weight of an embedded coalition $(S,P)$ depends solely on $|S|$ and $|P|$; this allows us to group all colorings that use the same number of colors.
\item Finally, in the last two values, the SS-value and the MY-value, weights depend on sizes of all coalitions in a partition.
\end{itemize}

Before we move to the next section, let us roughly explain a technique that we use in the proofs of \cref{theorem:ef_weighted,theorem:hy_embedded,theorem:my_embedded}. This technique was used in several complexity results for the Shapley value in games without externalities (see, e.g., \cite{Aziz:etal:2009:power,Michalak:etal:2013:connectivity}).

Assume we want to compute $x_1, \dots, x_k$ and we have an algorithm that computes the sum $f(j) = \sum_{m=1}^k a_{j,m} x_m$ for some weights $(a_{j,m})_{1 \le j,m \le k}$ that depend on $m$ and some external parameter $j \in \{1,\dots,k\}$. 
To this end, we can construct a system of linear equations with the following matrix form:
\begin{equation}\label{eq:technique}
\begin{bmatrix}
a_{1,1} & a_{1,2} & \dots & a_{1,k} \\
a_{2,1} & a_{2,2} & \dots & a_{2,k} \\
\vdots & \vdots & \ddots & \vdots \\
a_{k,1} & a_{k,2} & \dots & a_{k,k} \\
\end{bmatrix} 
\cdot 
\begin{bmatrix}
x_1 \\
x_2 \\
\vdots \\
x_k
\end{bmatrix} \\
=
\begin{bmatrix}
f(1) \\
f(2) \\
\vdots \\
f(k)
\end{bmatrix}
\end{equation}
Now, if the matrix $(a_{j,m})_{1 \le j,m \le k}$ has non-zero determinant, then it is invertible. 
Hence, if we know $f(1), \dots, f(k)$, then using Gaussian elimination we can compute $x_1, \dots, x_k$.

In our case, $f(j)$ will be an extended Shapley value and $x_m$ will be the number of independent sets of size $m$ (\cref{theorem:ef_weighted}), $k$-colorings that use $m$ colors (\cref{theorem:hy_embedded}) or matchings in a bipartite graph of size $m$ (\cref{theorem:my_embedded}). 
Hence, based on the fact that computing $\sum_{m=1}^k x_m$ is \#P-hard we will get that computing the EF-value, the HY-value and the MY-value is also \#P-hard. 

Note that based on \cref{eq:esv_is_colorings} each extended Shapley value is a sum over exponentially many colorings and, in general, two colorings that result in different partitions of nodes may have different weights.
Hence, the main challenge with this approach is to (1) express an extended Shapley value as a weighted sum over polynomial number of elements and (2) to create a system of linear equations that results in a matrix which is invertible.

%%%%%%%%%%%%%%%%%%%%%%%%%%%%%%%%%%%%%%%%%%%%%%%%%%%%%%%%%%%
%%%%%%%%%%%%%%           MCQUILLIN           %%%%%%%%%%%%%%
%%%%%%%%%%%%%%%%%%%%%%%%%%%%%%%%%%%%%%%%%%%%%%%%%%%%%%%%%%%
\subsection{Computing the MQ-value}
We begin with the analysis of the MQ-value.
Combining \cref{eq:mq_value,eq:lemma:game_from_g_gamma} as in \cref{eq:esv_is_colorings} we get:
\[ MQ_i(g^{\gamma}) = \sum_{f \in C_k(G^{\gamma})}  \frac{\zeta_i(S_f^*)}{\#_f} [|P_f| \leq 2].\]
As we can see, only colorings that use $1$ or $2$ colors have non-zero weights. 
All such colorings that result in the same partition of players can be grouped and their weights sum up to $\zeta_i(S^*_f)$. 
Hence, we can go through all $2$-colorings instead of $k$-colorings.
By observing that there are two $2$-colorings that result in the same partition we get:
\begin{equation}\label{eq:mq_final}
MQ_i(g^{\gamma}) = \frac{1}{2} \sum_{f \in C_2(G^{\gamma})} \zeta_i(S_f^*).
\end{equation}

While in a connected graph there are at most two $2$-colorings, in a disconnected graph it can be exponentially many.
Nevertheless, in \cref{theorem:mq_weighted}, we show that this sum can be easily computed in polynomial time for every graph.

\begin{theorem}\label{theorem:mq_weighted}
For a game represented as weighted MC-nets, the MQ-value can be computed in polynomial time.
\end{theorem}

\begin{proof}
From \cref{lemma:hybrid_rules_weighted} we know that it is enough to show that the MQ-value can be computed in polynomial time for a game represented as a single hybrid rule.

Fix a hybrid rule $\gamma$.
Assume $i \in l(u)$ for some $u \in V$.
From \cref{eq:mq_final} we have:
\[
MQ_i(g^{\gamma}) = \sum_{\substack{f \in C_2(G^{\gamma})\\ f(v^*) = f(u)}} \frac{(|S_f^*|-1)!(n-|S_f^*|)!}{2(n!)} - \sum_{\substack{f \in C_2(G^{\gamma})\\ f(v^*) \neq f(u)}} \frac{|S_f^*|!(n-|S_f^*|-1)!}{2(n!)}.
\]
Let us define two tables $T_{=}[1{\dots}n]$ and $T_{\neq}[1{\dots}n]$ as follows:
\[ T_{=}[s] = |\{f \in C_2(G^{\gamma}) : |S_f^*| = s, f(v^*) = f(u) \}|. \]
\[ T_{\neq}[s] = |\{f \in C_2(G^{\gamma}) : |S_f^*| = s, f(v^*) \neq f(u) \}|. \]
To put it in words, for $s \in \{1,\dots,n\}$, $T_{=}[s] + T_{\neq}[s]$ is the number of $2$-colorings in which there are $s$ players in nodes colored with the same color as node $v^*$. 
Now, $T_{=}[s]$ counts only these colorings in which $u$ is colored with the same color as $v^*$, and $T_{\neq}[s]$---with different color than $v^*$.
See Figure~\ref{figure:mq_example} for an illustration.
We have now:
\[ 
MQ_i(g^{\gamma}) = \sum_{s=1}^n \frac{(s-1)!(n-s)!}{2n!} T_{=}[s] - \sum_{s=1}^n \frac{s!(n-s-1)!}{2n!} T_{\neq}[s].
\]

\begin{figure}[t]
\centering
\begin{tikzpicture}[x=5cm,y=5cm]
  \node[main node,color1,minimum size=1.01cm] (1) at (0.00,1.12) {};
  \node[main node,color1] (1) at (0.00, 1.12) {$1,3$}; 
  \node[main node,color2] (2) at (0.00, 0.80) {$5$}; 
  \node[main node,color1] (3) at (0.40, 0.80) {$2$}; 
  \node[main node,color2] (4) at (0.55, 1.11) {$6$}; 
  \node[main node,color1] (5) at (0.70, 0.79) {$4,8$}; 
  \node[main node,color1] (6) at (1.10, 1.12) {$7$}; 

  \path[draw,thick] 
  (1) edge (2) 
  (3) edge (4) 
  (4) edge (5)
  ;
\end{tikzpicture}
\caption{There are eight $2$-colorings of this graph. Fix $i=6$. 
We have $T_{=} = [0,0,2,2,0,0,0,0]$ and $T_{\neq} = [0,0,0,0,2,2,0,0]$. 
For the presented coloring $f$ it holds $|S_f^*|=6$ and $f(v^*) \neq f(u)$ for $u$ s.t. $i \in l(u)$, hence $f$ is counted in $T_{\neq}[6]$.}
\label{figure:mq_example}
\end{figure}

It remains to determine tables $T_{=}$ and $T_{\neq}$.
Let us state some basic facts about $2$-colorings in a graph.
Graph is $2$-colorable if and only if it is bipartite, i.e., nodes can be partitioned into two groups $V = V_1 \dot\cup V_2$, such that $V_1$ and $V_2$ are independent sets.
If a bipartite graph is connected (i.e., there exists a path between any two nodes), then there exist a unique such partition. 
It can be found by performing a breadth-first search from any node, $v \in V$, and putting all nodes at even distance from $v$ in set $V_1$ and all nodes at odd distance---in set $V_2$. 
Note that creating such a partition and checking whether both groups are independent sets is also  a good way of checking whether the graph is $2$-colorable.
Now, there are two $2$-colorings: in the first one nodes from $V_1$ are colored with color $1$, and in the second one---with color $2$.

On the other hand, if the graph is not connected, then it is $2$-colorable if its every connected component (i.e., maximal subset of nodes such that there exists a path between every pair of nodes) is $2$-colorable. 
In such a case, for each connected component there exists a unique partition into independent sets. 
However, in the whole graph there may be an exponential number of $2$-colorings.

Let $C_1, \dots, C_m$ be connected components of graph $G^{\gamma}$, and let $\{A_j, B_j\}$ be a partition of $C_j$ into two independent sets.
If partition $\{A_j, B_j\}$ does not exist for at least one component $C_j$, then there are no 2-colorings and $T_{=}$ and $T_{\neq}$ have only zeros.
Assume otherwise and without loss of generality assume $v^* \in A_1$.
Let $T_{=}$ and $T_{\neq}$ be filled with zeros. 
We initiate both tables depending on the position of node $u$:
\begin{itemize}
\item if $u \in A_1$, then $T_{=}[|l(A_1)|] = 2$;
\item if $u \in B_1$, then $T_{\neq}[|l(A_1)|] = 2$;
\item otherwise, without loss of generality let us assume that $u \in A_2$; hence, $T_{=}[|l(A_1)|+|l(A_2)|] = T_{\neq}[|l(A_1)|+|l(B_2)|] = 2$.
\end{itemize}
Now, we consider other components $(C_2), C_3, \dots, C_m$, one by one, and for each component $C_j$ consider two cases: either $A_j$ or $B_j$ is colored with the same color as node $v_1$. 
Thus, in each step, we update each table $T_{*}$ ($T_{=}$ or $T_{\neq}$) by replacing it with a new table $T'_{*}$ defined as follows:
\begin{equation}\label{eq:proof_mq_a}
T'_{*}[s] = T_{*}[s-|l(A_j)|] + T_{*}[s-|l(B_j)|] \text{ for }1 \le s \le n,
\end{equation}
assuming $T_{*}[s] = 0$ for $s \le 0$. 
After analyzing the $m$-th component, the calculation is complete.
This concludes the proof.
\end{proof}

\cref{theorem:mq_weighted} implies polynomial computation also for embedded MC-nets.

\begin{corollary}\label{corollary:mq_embedded}
For a game represented as embedded MC-nets, the MQ-value can be computed in polynomial time.
\end{corollary}
\begin{proof}
Directly from \cref{theorem:mq_weighted}. This result was also proved by \citeA{Michalak:etal:2010:embedded}.
\end{proof}

%%%%%%%%%%%%%%%%%%%%%%%%%%%%%%%%%%%%%%%%%%%%%%%%%%%%%%%%%%%
%%%%%%%%%%%%%%      EXTERNALITY-FREE         %%%%%%%%%%%%%%
%%%%%%%%%%%%%%%%%%%%%%%%%%%%%%%%%%%%%%%%%%%%%%%%%%%%%%%%%%%
\subsection{Computing the EF-value}
The EF-value, considered by us next, is complementary to the MQ-value.
Combining \cref{eq:ef_value,eq:lemma:game_from_g_gamma} gives:
\[ EF_i(g^{\gamma}) =  \sum_{f \in C_k(G^{\gamma})} \frac{\zeta_i(S_f^*)}{\#_f} \cdot [|P_f|-1 = n - |S_f^*|]. \]
Note that condition $|P_f|-1 = n-|S_f^*|$ holds if and only if every player $i \in N \setminus S_f^*$ form a singleton coalition $\{i\}$ in $P_f$, i.e., $P_f = \{S_f^*\} \cup \{\{i\} : i \in N \setminus S_f^*\}$.
This means that every node $u \in V \setminus VS_f^*$ has a label of size one and is colored with a different color than all other nodes.
Hence, $VP_f$ is uniquely defined by set $VS_f^*$.

Let us split nodes into three sets:
\begin{itemize}
\item $V' = \{v \in V : |l(v)|>1\} \cup \{v^*\}$ consists of all nodes with non-singleton labels and node $v^*$;
\item $NV' = \left(\bigcup_{v \in V'} \mathcal{N}(v)\right) \setminus V'$ consists of all neighbors of nodes from $V'$;
\item $U' = V \setminus (V' \cup NV')$ consists of all the remaining nodes.
\end{itemize}
See Figure~\ref{figure:eq_proofs} for an illustration.
Since nodes with non-singleton labels cannot be in $V \setminus VS_f^*$, we know that $VS_f^*$ contains $V'$. 
Hence, if $V'$ is not an independent set, then the formula evaluates to zero.
Assume otherwise.
We know that $VS_f^*$ is an independent set, hence it cannot contain any node from $NV'$. 
As a result, $VS_f^*$ is the union of $V'$ and a subset of $U'$ that is an independent set.

\begin{figure}[t]
\centering
\begin{tikzpicture}[x=5cm,y=5cm] % change these values to adjust the size of a figure
  \node[main node,minimum size=1.01cm] (1) at (0.00, 0.99) {}; 
  \node[main node,color2] (1) at (0.00, 0.99) {$1$}; 
  \node[main node,color2] (2) at (0.20, 0.78) {$4$,$6$}; 
  \node[main node,color2] (3) at (0.20, 1.18) {$3$,$7$}; 
  \node[main node,color1] (5) at (0.60, 1.12) {$2$}; 
  \node[main node,color4] (8) at (1.00, 1.18) {$5$}; 
  \node[main node,color3] (6) at (0.60, 0.81) {$8$}; 
  \node[main node,color2] (7) at (1.00, 0.88) {$9$}; 
  
  \draw[dashed] (-0.15,0.65) rectangle (0.35, 1.33);
  \node at (0.1, 0.58) {$V'$};

  \draw[dashed] (0.45,0.65) rectangle (0.72, 1.33);
  \node at (0.6, 0.58) {$NV'$};

  \draw[dashed] (0.86,0.65) rectangle (1.13, 1.33);
  \node at (!.0, 0.58) {$U'$};

  \path[draw,thick]
  (7) edge  (6)
  (7) edge  (5)
  (2) edge  (5)
  (1) edge  (5)
  (8) edge  (7)
  (1) edge  (6)
  (3) edge  (5)
  (5) edge  (6)
  ;
\end{tikzpicture}
\caption{Partition of nodes into three sets: $V'$, $NV'$ and $U'$ and an example of a coloring $f$ in which $|P_f|-1=n-|S_f^*|$: nodes from $V'$ have the same color, each other color appears at most once in a node with a singleton label.}
\label{figure:eq_proofs}
\end{figure}

Also, if $VS = VS_f^*$ for some coloring $f$, then there are $\#_f$ colorings with the same set $VS_f^*$.
Hence, we get the following formula for the EF-value:
\begin{equation}\label{eq:ef_final}
EF_i(g^{\gamma}) = \sum_{U \in I(G^{\gamma}) : U \subseteq U'} \zeta_i(l(V' \cup U))
\end{equation}
where $I(G^{\gamma})$ is the set of all independent sets in graph $G^{\gamma}$.

In the following two theorems, we show that this sum is hard to compute in general, but it is easy to compute if the player-graph is regular.

\begin{theorem}\label{theorem:ef_embedded}
For a game represented as embedded MC-nets, the EF-value can be computed in polynomial time.
\end{theorem}

\begin{proof}
Let $\gamma$ be a regular hybrid rule and $G^{\gamma}$ be the corresponding regular player-graph (from Lemmas~\ref{lemma:hybrid_rules_embedded} and \ref{lemma:g_gamma_embedded} we know that for every embedded MC-nets rule such equivalent hybrid rule and graph exists).

Let us analyze graph $G^{\gamma}$.
First, observe that if there exists a node, other than $v^*$, with non-singleton label, then since $G^{\gamma}$ is regular it is adjacent to $v^*$.
Hence, $V'$ does not form an independent set and, as we already argued, the formula evaluates to zero: $EF_i(g^{\gamma}) = 0$ for every $i \in N$. 

Assume otherwise, i.e., that all nodes other than $v^*$ have singleton labels: $V' = \{v^*\}$.
Since $G^{\gamma}$ is regular, we get that $U'$ is an independent set.
Thus, from \cref{eq:ef_final}:
\[ EF_i(g^{\gamma}) = \sum_{U \subseteq U'} \zeta_i(l(U \cup \{v^*\})). \]				
Note that $|l(v^*)| = n-k+1$ and $|l(U \cup \{v^*\})| = |U| + n-k+1$ where $k$ is the number of nodes in graph $G^{\gamma}$.
Therefore:
\begin{itemize}
\item If $i \in l(v^*)$, then:
\[ EF_i(g^{\gamma}) = \sum_{s=0}^{|U'|} \binom{|U'|}{s} \frac{(s+n-k)!(k-s-1)!}{n!} = \frac{(n-k)!(k-|U'|-1)}{(n-|U'|)!}.\]
\item If $i \in N \setminus l(U' \cup \{v^*))$, then 
\[ EF_i(g^{\gamma}) = - \sum_{s=0}^{|U'|} \binom{|U'|}{s} \frac{(s+n-k+1)!(k-s-2)!}{n!} = \frac{(n-k+1)!(k-|U'|-2)}{(n-|U'|)!}.\]
\item If $i \in l(U')$, then:
\begin{multline*} 
EF_i(g^{\gamma}) = \sum_{s=1}^{|U'|} \binom{|U'|\!-\!1}{s-1} \frac{(s+n-k)!(k-s-1)!}{n!} \\
- \sum_{s=0}^{|U'|-1} \binom{|U'|\!-\!1}{s-1} \frac{(s+n-k+1)!(k-s-2)!}{n!} = 0.
\end{multline*}
\end{itemize}
These values can be computed in polynomial time.
This concludes the proof.
\end{proof}

%\begin{proof}[Sketch of proof]
%From \cref{lemma:hybrid_rules_embedded} it is enough to consider regular hybrid rules with weight $1$.
%Fix such a hybrid rule $\gamma$ and consider $G^{\gamma}$.
%If $V^* \setminus \{v_1\} \neq \emptyset$, i.e., there exists a node, other than $v_1$, with the size of a label larger than one, then from \cref{lemma:g_gamma_embedded} it must be adjacent to $v_1$; hence, $EF_i(g^{\gamma}) = 0$ for every $i \in N$. 
%Assume otherwise. 
%We get that $V^* \cup \{v_1\} = \{v_1\}$, and $U$ is the set of nodes not adjacent to $v_1$. 
%From \cref{lemma:g_gamma_embedded} we know that $U$ is an independent set. 
%Thus, $EF_i(g^{\gamma}) = \sum_{S \subseteq U} \zeta_i(l(S \cup \{v_1\}))$ which can be computed in polynomial time.
%\end{proof}

\begin{theorem}\label{theorem:ef_weighted}
For a game represented as weighted MC-nets, computing the EF-value is \#P-hard.
\end{theorem}

\begin{proof}
%The value $EF_i(g^{\gamma})$ can be considered as the number of accepting paths of nondeterministic Turing machine, so the problem is in \#P.
%To show that the problem is \#P-hard, 
We use a Turing reduction from the problem of counting all independent sets in a graph which is \#P-complete \cite{Valiant:1979:enumeration}. 

Let $G = (V,E)$ be an arbitrary graph with $V = \{v_1,\dots,v_k\}$.
Let $I_m(G)$ be the set of independent sets of size $m$ in graph $G$.
We will determine $|I_m(G)|$ for every $m \in \{0,\dots,k\}$.

\begin{figure}[t]
\centering
\begin{tikzpicture}[x=5cm,y=5cm] % change these values to adjust the size of a figure
  \node[main node,label={left:$v_1$}] (1) at (0.17, 1.00) {$1$}; 
  \node[main node,label={left:$v_2$}] (2) at (0.00, 0.80) {$2$}; 
  \node[main node,label={below:$v_3$}] (3) at (0.31, 0.72) {$3$}; 
  \node[main node] (4) at (0.51, 1.00) {\ $\dots$}; 
  \node[main node,label={right:$v_k$}] (5) at (0.63, 0.73) {$k$}; 
  \node[main node,label={left:$v^*$}] (6) at (1.30, 0.85) {};
%  \node[main node,minimum size=1.02cm] (6) at (1.30, 0.85) {};

  \draw[->,thick] (1.30,0.85) -- (1.45,0.81);
  \node[draw=none] (6a) at (1.91, 0.80) {$l(v^*) = k$+1, $\dots$, $k$+$j$+1};
  
  \draw[dashed] (0.33,0.84) ellipse (2.8cm and 1.8cm);
  \node[draw=none] at (0.86,1.1) {$G$};

  \path[draw,thick]
  (4) edge  (3)
  (1) edge  (4)
  (2) edge  (1)
  (1) edge  (3)
  (3) edge  (2)
  (5) edge  (4)
  (5) edge  (1)
  ;
\end{tikzpicture}
\caption{Graph $G^{\gamma_j}$ from the proof of \cref{theorem:ef_weighted}.} 
\label{figure:ef_weighted}
\end{figure}

To this end, let us construct $k+1$ player-graphs: 
for $j \in \{0, \dots, k\}$ we label each node $v_i$ with $i$ and add an isolated node $v^*$ with label of size $j+1$.
Specifically, we construct a player-graph $G^{\gamma_j} = (V \cup \{v^*\}, E,l,v^*)$ with $l(v_i) = \{i\}$ for $v_i \in V$ and $l(v^*) = \{k+1, \dots, k+j+1\}$ (see \cref{figure:ef_weighted} for an illustration).
Based on \cref{lemma:g_gamma_weighted} we know that there exists a hybrid rule $\gamma_j$ for which this graph is $G^{\gamma_j}$.
Note that in this graph we have $V' = \{v^*\}$ and $U' = V$. 
Hence, \cref{eq:ef_final} for $i = k+1$ and graph $G^{\gamma_j}$ simplifies to:
\[ EF_i(g^{\gamma_j}) = \sum_{m=0}^{k} \frac{(m+j)!(k-m)!}{(k+j+1)!} |I_m(G)|. \]
This system of linear equations is equivalent to the following matrix form:
\[
\begin{bmatrix}
0!k! & 1!(k-1)! & \dots  & k!0! \\
1!k! & 2!(k-1)! & \dots  & (k+1)!0! \\
\vdots & \vdots & \ddots & \vdots \\
k!k! & (k+1)!(k-1)! & \dots  & (2k)!0!
\end{bmatrix} 
\cdot 
\begin{bmatrix}
|I_0(G)| \\
|I_1(G)| \\
\vdots \\
|I_k(G)|
\end{bmatrix}
=
\begin{bmatrix}
(k+1)! EF_i(g^{\gamma_0}) \\
(k+2)! EF_i(g^{\gamma_1}) \\
\vdots \\
(2k+1)! EF_i(g^{\gamma_k})
\end{bmatrix}
\]
From \cite[Theorem 1.1]{Bacher:2002}, we know that the determinant of matrix $A = ((a+b)!)_{0 \le a,b \le k}$ equals $\prod_{a=0}^{k} (a!)^2$; hence, the determinant of the above square matrix is $\prod_{a=0}^{k} (a!)^3$ (columns of $A$ are multiplied by $k!, \dots, 0!$). 
Since the determinant is non-zero, the matrix is invertible and knowing $EF_i(g^{\gamma_0}), \dots, EF_i(g^{\gamma_k})$ allows us to find $I_0(G), \dots, I_k(G)$ in polynomial time using Gaussian elimination. 
Hence, computing the EF-value is \#P-hard.
\end{proof}

%%%%%%%%%%%%%%%%%%%%%%%%%%%%%%%%%%%%%%%%%%%%%%%%%%%%%%%%%%%
%%%%%%%%%%%%%%           HU-YANG             %%%%%%%%%%%%%%
%%%%%%%%%%%%%%%%%%%%%%%%%%%%%%%%%%%%%%%%%%%%%%%%%%%%%%%%%%%
\subsection{Computing the HY-value}
The HY-value is the first value considered by us with non-zero weights of every embedded coalition. 
Here, value $g(S,P)$ is multiplied by $\zeta_i(S) \cdot \psi(S,P)$, where $\psi(S,P) = |\{P' \in \mathcal{P} : \{T \setminus S : T \in P'\} = P \setminus \{S\}\}|$. 
To put it in words, $\psi(S,P)$ is the number of partitions that can be obtained from $P \setminus \{S\}$ by inserting players from $S$. 

Let us express values $\psi(S,P)$ using a notion of generalized \emph{Bell numbers}.
The \emph{$n$-th Bell number}, denoted by $B_n$, is the number of all possible partitions of $n$ elements. 
Now, \emph{$r$-Bell numbers} are a generalization of Bell numbers: $B_{n,r}$ is the number of partitions of $n+r$ elements such that the first $r$ elements are in distinct subsets~\cite{Mezo:2011}. In particular, $B_{1,2} = 3$: we have $\{\{1,3\}, \{2\}\}, \{\{1\}, \{2,3\}\}, \{\{1\}, \{2\}, \{3\}\}$.

Now, observe that $\psi(S,P) = B_{|S|,|P|-1}$.
Thus, \cref{eq:hy_value,eq:lemma:game_from_g_gamma} yields:
\begin{equation}\label{eq:hy_final}
HY_i(g^{\gamma}) = \frac{1}{B_n} \sum_{f \in C_k(G^{\gamma})} \frac{\zeta_i(S_f^*)}{\#_f} \cdot B_{|S_f^*|,|P_f|-1}.
\end{equation}
Thus, for a fixed player $i$ and the size of $S_f^*$, the weight of a coloring depends solely on the number of colors it uses.

We will prove that computing this sum is \#P-hard.
To this end, we first identify the determinant of matrix of generalized Bell numbers.

\begin{lemma}\label{lemma:bnr_matrix}
The determinant of matrix $B = (B_{j,m})_{1 \le j,m \le k}$ equals $(\prod_{i=0}^k i!) \cdot (\sum_{i=0}^k 1/i!)$.
\end{lemma}

\begin{proof}
Matrix $B$ looks as follows:
\[ B = \begin{bmatrix}
B_{1,1} & B_{1,2} & \dots & B_{1,k} \\
B_{2,1} & B_{2,2} & \dots & B_{2,k} \\
\vdots & \vdots & \ddots & \vdots \\
B_{k,1} & B_{k,2} & \dots & B_{k,k}
\end{bmatrix} \]
Let $S(n,k)$ be the Stirling number of the second kind, i.e., the number of partitions of $n$ elements into $k$ subsets (e.g., $S(4,2) = 7$).
Consider matrix:
\[ C = \begin{bmatrix}
S(1,1) & S(2,1) & S(3,1) & \dots  & S(k,1) \\
0 & S(2,2) & S(3,2) & \dots  & S(k,2) \\
0 & 0 & S(3,3) & \dots  & S(k,3) \\
\vdots & \vdots & \vdots & \ddots & \vdots \\
0 & 0 & 0 & \dots  & S(k,k)
\end{bmatrix} \]
Since $S(i,i) = 1$ for every $i \in \mathbb{N}$, matrix $C$ is triangular with diagonal $1$, so $\det(C) = 1$. 
Consider the product $A = B \cdot C$.
Multiplying by $C$ is equivalent to adding to column $j$ columns $1,2,\dots,j-1$ with weights $S(j,1), S(j,2), \dots, S(j,j-1)$.
Hence, we get that:
\[ A[i,j] = B_{i,1} \cdot S(j,1) + \dots + B_{i,j} \cdot S(j,j). \] 
Let us argue that $A[i,j] = B_{i+j}$.
Take $j$ first elements from $i+j$ and consider all their possible partitions.
There are $S(j,m)$ partitions of $j$ elements into $m$ subsets which implies there are $S(j,m) B_{i,m}$ partitions of $i+j$ elements in which first $j$ elements form $m$ subsets. 
Summing over all $m \in \{1,\dots,j\}$ we get all $B_{i+j}$ partitions of $i+j$ elements, each exactly once.

As a result, we get that $A = B \cdot C = (B_{i+j})_{1 \le i,j \le k}$.
Now, from \cite[Remark 2]{Aigner:1999} we get that the determinant of matrix $(B_{i+j})_{1 \le i,j \le k}$ is $(\sum_{i=0}^k k!/(k-i)!) \cdot (\prod_{i=0}^{k-1} i!)$ which is equivalent to $(\prod_{i=0}^{k} i!) \cdot (\sum_{i=0}^k 1/i!)$. Hence, it is also the determinant of matrix $B$.
\end{proof}

%\begin{lemma}\label{lemma:bnr_recursion}
%For every $n,r \in \mathbb{N}$, $B_{n+r} = \sum_{i=1}^r S(r,i) B_{n,i}$, where $S(r,i)$ is the Stirling number of the second kind, i.e., the number of partitions of $r$ elements into $i$ subsets.
%\end{lemma}
%
%\begin{proof}
%Take $r$ first elements from $n+r$ and consider all their possible partitions. There are $S(r,i)$ partitions of $r$ elements into $i$ subsets which implies there are $S(r,i) B_{n,i}$ partitions of $n+r$ elements in which first $r$ elements form $i$ subsets. Summing over all $i \in \{1,\dots,r\}$ concludes the proof.
%\end{proof}

We are now ready to state the main theorem.

\begin{theorem}\label{theorem:hy_embedded}
For a game represented as embedded MC-nets or weighted MC-nets, computing the HY-value is \#P-hard.
\end{theorem}

\begin{proof}
%The value $HY_i(g^{\gamma})$ can be considered as the number of accepting paths of nondeterministic Turing machine, so the problem is in \#P.
%To show that the problem is \#P-hard, 
We use a Turing reduction from the \emph{chromatic polynomial problem}, i.e., counting $m$-colorings in a graph, which is \#P-complete~\cite{Jaeger:etal:1990}.

Let $G = (V,E)$ be an arbitrary graph with $V = \{v_1,\dots,v_k\}$ and $m$ be an arbitrary number. 
The task is to determine $|C_m(G)|$.
Let $c_i$ be the number of $k$-colorings that uses exactly $i$ colors.
We will determine $c_i$ for every $i \in \{1,\dots,k\}$.
From these values it is easy to compute the number of $m$-colorings with the following formula: 
\[ |C_m(G)| = \sum_{i=1}^{k} \frac{\binom{m}{i}}{\binom{k}{i}} c_i = \sum_{i=1}^{k} \frac{m(m-1)\cdots(m-i+1)}{k(k-1)\cdots(k-i+1)} \cdot c_i. \]
The argument is as follows: 
Consider an $i$-coloring $f$ that uses $i$-colors: $\{1,\dots,i\}$.
Now, every $k$-coloring that uses exactly $i$ colors corresponds to one of these colorings. 
Specifically, there are $\binom{k}{i}$ $k$-colorings that preserves the same partition of nodes and order of colors, i.e., $g: V \rightarrow \{1,\dots,k\}$ s.t. $g(v_i) \le g(v_j)$ if and only if $f(v_i) \le f(v_j)$.
Analogically, there are $\binom{m}{i}$ such $m$-colorings.
Hence, if we have the number of $k$-colorings that uses $i$ colors to obtain the number of $m$-colorings that uses $i$ colors we need to divide by $\binom{k}{i}$ and multiply by $\binom{m}{i}$.

Let us construct $k$ graphs: for $j \in \{1,\dots,k\}$ we label each node $v_i$ with $i$ and add node $v^*$ connected to all nodes from $V$ with label of size $j$.
Specifically, we construct a player-graph $G^{\gamma_j} = (V \cup \{v^*\}, E \cup \{\{v_i,v^*\} : v_i \in V\}, l, v^*)$ with $l(v_i) = \{i\}$ for $v_i \in V$ and $l(v^*) = \{k+1, \dots, k+j\}$. 
Since node $v^*$ is connected to all other nodes we know that $G^{\gamma_j}$ is regular, i.e., there exists a regular hybrid rule $\gamma_j$ equivalent to some embedded MC-nets rule such that $G^{\gamma_j}$ is the corresponding graph. 

\begin{figure}[t]
\centering
\begin{tikzpicture}[x=5cm,y=5cm] % change these values to adjust the size of a figure
  \node[main node,label={left:$v_1$}] (1) at (0.17, 1.00) {$1$}; 
  \node[main node,label={left:$v_2$}] (2) at (0.00, 0.80) {$2$}; 
  \node[main node,label={below:$v_3$}] (3) at (0.31, 0.72) {$3$}; 
  \node[main node] (4) at (0.51, 1.00) {\ $\dots$}; 
  \node[main node,label={below:$v_k$}] (5) at (0.63, 0.73) {$k$}; 
  \node[main node,label={above:$v^*$}] (6) at (1.30, 0.85) {};
%  \node[main node,minimum size=1.02cm] (6) at (1.30, 0.85) {};

  \draw[->,thick] (1.30,0.85) -- (1.45,0.81);
  \node[draw=none] (6a) at (1.86, 0.80) {$l(v^*) = k$+1, $\dots$, $k$+$j$};
  
  \draw[dashed] (0.33,0.84) ellipse (2.8cm and 1.8cm);
  \node[draw=none] at (0.86,1.1) {$G$};

  \path[draw,thick]
  (4) edge  (3)
  (1) edge  (4)
  (2) edge  (1)
  (1) edge  (3)
  (3) edge  (2)
  (5) edge  (4)
  (5) edge  (1)
  ;
  
  \path[draw,thick,gray]
  (6) edge  (1)
  (6) edge  (2)
  (6) edge  (3)
  (6) edge  (4)
  (6) edge  (5)
  ;
\end{tikzpicture}
\caption{Graph $G^{\gamma_j}$ from the proof of \cref{theorem:hy_embedded}.} 
\label{figure:hy_embedded}
\end{figure}

Let us analyze the HY-value of player $i = k+1$ in game $g^{\gamma_j}$.
We have $k+1$ nodes in graph $G^{\gamma_j}$, $n = k+j$ and $|S_f^*| = j$ for every coloring $f$.
Hence, Equation~\eqref{eq:hy_final} yields:
\begin{eqnarray*}
HY_i(g^{\gamma_j}) & = & \frac{(j-1)!k!}{(k+j)! B_{k+j}} \sum_{f \in C_{k+1}(G^{\gamma})} \frac{B_{j,|P_f|-1}}{\#_f} \\
& = & \frac{(j-1)!k!}{(k+j)! B_{k+j}} \sum_{m=1}^{k} \sum_{f \in C_{k+1}(G^{\gamma}) : |P_f| = m+1} \frac{(k-m)!}{(k+1)!} B_{j,m}
\end{eqnarray*}
Here, we used the fact that if $(k+1)$-coloring $f$ uses $m+1$ colors, then $\#_f = (k+1)!/(k-m)!$.
The number of $(k+1)$-colorings of graph $G^{\gamma_j}$ that uses $m+1$ colors is equal to $(k+1)$ (color of node $v^*$) times the number of $k$-colorings of graph $G$ that uses $m$ colors:
\[ |\{f \in C_{k+1}(G^{\gamma}) : |P_f| = m+1\}| = (k+1) \cdot c_m. \]
Hence, we get:
\[
HY_i(g^{\gamma}) = \frac{(j-1)!}{(k+j)!B_{k+j}} \sum_{m=1}^{k} (k-m)! \cdot B_{j,m} \cdot c_m.
\]
This system of linear equations can be presented in the matrix form as follows:
\[
\begin{bmatrix}
(k-1)! B_{1,1} & (k-2)! B_{1,2} & \dots  & 0! B_{1,k} \\
(k-1)! B_{2,1} & (k-2)! B_{2,2} & \dots  & 0! B_{2,k} \\
\vdots & \vdots & \ddots & \vdots \\
(k-1)! B_{k,1} & (k-2)! B_{k,2} & \dots  & 0! B_{k,k}
\end{bmatrix} 
\cdot 
\begin{bmatrix}
c_1 \\
c_2 \\
\vdots \\
c_k
\end{bmatrix}
=
\begin{bmatrix}
\frac{(k+1)!}{0!} B_{k+1} HY_i(g^{\gamma_1}) \\
\frac{(k+2)!}{1!} B_{k+2} HY_i(g^{\gamma_2}) \\
\vdots \\
\frac{(2k)!}{(k-1)!} B_{2k} HY_i(g^{\gamma_k})
\end{bmatrix}
\]
From \cref{lemma:bnr_matrix} we get that the determinant of the square matrix equals $(\prod_{i=0}^k i!) \cdot (\sum_{i=0}^k 1/i!)$ multiplied by $(\prod_{i=0}^{k-1} i!)$.
Since the determinant is non-zero, the matrix is invertible and knowing $HY_i(g^{\gamma_1}), \dots, HY_i(g^{\gamma_k})$ allows us to find $c_1, \dots, c_k$ in polynomial time using Gaussian elimination. 
This concludes the proof.
\end{proof}

\subsection{Computing the SS-value}
The SS-value, considered next, is probably the most popular extended Shapley value. 
Combining \cref{eq:ss_value,eq:lemma:game_from_g_gamma} gives:
\[ SS_i(g^{\gamma}) = \sum_{f \in C_k(G^{\gamma})} \frac{\zeta_i(S_f^*)}{\#_f} \cdot \frac{\prod_{T \in P_f \setminus \{S_f^*\}} (|T|-1)!}{(n-|S_f^*|)!}. \]

In what follows, let us focus on graphs in which every node is labeled with a single player: $|l(v)| = 1$ for every $v \in V$.
In such a case, we have $n=k=|V|$ and $|l(T)| = |T|$. 
Under this assumption, formula for the SS-value of player $i \in l(v^*)$ is as follows:
\begin{equation}\label{eq:ss_colorings}
SS_i(g^{\gamma}) = \frac{1}{n!} \sum_{f \in C_k(G^{\gamma})} \frac{\prod_{T \in P_f} (|T|-1)!}{\#_f}.
\end{equation}
As we already mentioned, for a fixed partition $P \in \mathcal{P}$, value $\prod_{T \in P} (|T|-1)!$ is the number of permutations in which $P$ is the partition obtained from a cycle decomposition (for a permutation $h: N \rightarrow N$ such partition is defined as follows: $\{\{i, h(i), h(h(i)), \dots\} : i \in N\}$).
Hence, according to \cref{eq:ss_colorings}, the SS-value of player $i$ is equal to the probability that the partition obtained from a cycle decomposition of a random permutation of nodes corresponds to a (proper vertex) coloring in graph $G^{\gamma}$.

Let us consider a complement of graph $G^{\gamma} = (V,E)$:
\[ \overline{G^{\gamma}} = (V, \{\{u,v\} : u,v \in V, u \neq v\} \setminus E). \]
For every coloring $f \in C_k(G^{\gamma})$, sets of nodes in $VP_f$ are independent sets in $G^{\gamma}$. 
Hence, they are cliques in $\overline{G^{\gamma}}$. 
As a consequence, we get that $SS_i(g^{\gamma})$ from \cref{eq:ss_colorings} is equivalently a weighted sum over \emph{clique covers} (i.e., partitions of the nodes in a graph into cliques):
\begin{equation}\label{eq:ss_final}
SS_i(g^{\gamma}) = \frac{1}{n!} \sum_{P \in QC(\overline{G^{\gamma}})} \prod_{T \in P} (|T|-1)!,
\end{equation}
where $QC(G)$ is the set of all clique covers in graph $G$.
In the following theorem we prove that computing this sum is \#P-hard.

\begin{theorem}\label{theorem:ss_embedded}
For a game represented as embedded MC-nets or weighted MC-nets, computing the SS-value is \#P-hard.
\end{theorem}

\begin{figure}[t]
\centering
\begin{tikzpicture}[x=5cm,y=5cm] % change these values to adjust the size of a figure
  \def\x{0.0}

  \node[main node] (1) at (\x+0.17, 1.00) {$1$}; 
  \node[main node] (2) at (\x+0.00, 0.72) {$2$}; 
  \node[main node] (3) at (\x+0.31, 0.72) {$3$}; 
  \node[main node] (4) at (\x+0.51, 1.00) {$\dots$}; 
  \node[main node] (5) at (\x+0.63, 0.72) {$k\!\shortminus\!1$}; 
  \node[main node,label={below:$v^*$}] (6) at (\x+1.02, 0.85) {$k$};

  \draw[dashed] (\x+0.33,0.84) ellipse (2.8cm and 1.8cm);
  \node[draw=none] at (\x+0.86,1.1) {$G$};

  \path[draw,thick]
  (4) edge  (3)
  (2) edge  (1)
  (1) edge[line width=2pt]  (3)
  (5) edge[line width=2pt]  (4)
  (5) edge  (1)
  ;

  \def\x{1.65}
  \node[main node,color1] (1) at (\x+0.17, 1.00) {$1$}; 
  \node[main node,color2] (2) at (\x+0.00, 0.72) {$2$}; 
  \node[main node,color1] (3) at (\x+0.31, 0.72) {$3$}; 
  \node[main node,color3] (4) at (\x+0.51, 1.00) {$\dots$}; 
  \node[main node,color3] (5) at (\x+0.63, 0.72) {$k\!\shortminus\!1$}; 
  \node[main node,color4,label={below:$v^*$}] (6) at (\x+1.02, 0.85) {$k$};

  \draw[dashed] (\x+0.33,0.84) ellipse (2.8cm and 1.8cm);
  \node[draw=none] at (\x+0.86,1.1) {$\overline{G}$};

  \path[draw,thick]
  (1) edge (4)
  (4) edge (2)
  (2) edge (3)
  (2) edge[bend right=35] (5)
  (3) edge (5)
  (6) edge[bend left=9] (1)
  (6) edge[out=-180,in=25] (2)
  (6) edge[bend right=11] (3)
  (6) edge (4)
  (6) edge (5)
  ;
  
\end{tikzpicture}
\caption{Graph $\overline{G^{\gamma}}$ (one the left) and $G^{\gamma}$ (on the right) from the proof of \cref{theorem:ss_embedded}. 
The matching highlighted on graph $\overline{G^{\gamma}}$ corresponds to the coloring from graph $G^{\gamma}$.} 
\label{figure:ss_embedded}
\end{figure}

\begin{proof}
%The value $SS_i(g^{\gamma})$ can be considered as the number of accepting paths of nondeterministic Turing machine, so the problem is in \#P.
We use a reduction from the problem of counting all matching in a bipartite graph which is \#P-complete~\cite{Valiant:1979:enumeration}.

Let $G = (V,E)$ be an arbitrary bipartite graph with $V = \{v_1,\dots,v_{k-1}\}$ for notational convenience.
Let $h_G$ be the number of all matchings in $G$ (called Hosoya index). 
Our goal is to determine $h_G$.

To this end, let us construct a graph $\overline{G^\gamma} = (V \cup \{v^*\}, E)$ with $l(v_i) = \{i\}$ for $v_i \in V$ and $l(v^*) = \{k\}$.
See Figure~\ref{figure:ss_embedded} for an illustration.
Consider a complement of graph $\overline{G^{\gamma}}$, denoted by $G^{\gamma}$.
Since node $v^*$ does not have any edges in $\overline{G^{\gamma}}$, then it is connected to all nodes in $G^{\gamma}$; hence, $G^{\gamma}$ is regular and there exists a hybrid rule $\gamma$ equivalent to some embedded MC-nets rule such that $G^{\gamma}$ is the corresponding graph.

Now, consider $SS_i(g^{\gamma})$ for $i=k$.
From \cref{eq:ss_final} we know that $SS_i(g^{\gamma})$ is a weighted sum over clique covers of $\overline{G^{\gamma}}$, i.e., partitions of nodes of $\overline{G^{\gamma}}$ into cliques. 
However, since the graph is bipartite, there exist no clique with more than $2$ nodes. 
Hence, for every such a partition, $P$, we have $\prod_{T \in P} (|T|-1)! = 1$.
Thus, from \cref{eq:ss_final} we get:
\[ SS_i(g^{\gamma}) =\frac{1}{k!} |QC(\overline{G^{\gamma}})| = \frac{1}{k!} |QC(G)| = \frac{1}{k!} h_G. \]
where the last equality comes from the fact that in a bipartite graph every partition of nodes into cliques corresponds to exactly one matching. This concludes the proof.
\end{proof}

%%%%%%%%%%%%%%%%%%%%%%%%%%%%%%%%%%%%%%%%%%%%%%%%%%%%%%%%%%%
%%%%%%%%%%%%%%           MYERSON             %%%%%%%%%%%%%%
%%%%%%%%%%%%%%%%%%%%%%%%%%%%%%%%%%%%%%%%%%%%%%%%%%%%%%%%%%%
\subsection{Computing the MY-value}
The last value that we consider is the MY-value which is the first chronologically proposed extension of the Shapley value.
Combining \cref{eq:my_value,eq:lemma:game_from_g_gamma} gives:
\[ MY_i(g^{\gamma}) = \! \sum_{f \in C_k(G^{\gamma})} \! \frac{(-1)^{|P_f|} (|P_f|\!-\!2)!}{\#_f} h_i(f), \mbox{ with } h_i(f) =  \frac{1\!-\!|P_f|}{n} + \sum_{\substack{T \in P_f \setminus \{S_f^*\}\\i \not \in T}} \frac{1}{(n\!-\!|T|)}.\]
We note that both techniques used for the HY-value and the SS-value does not work in this case.

To cope with this problem, we will exploit the fact that weights of the MY-value have a form of a sum over all coalitions.
Specifically, we will consider a difference between the MY-value of two players. 
Let us denote such difference for players $i$ and $i'$ in game $g$ by $MY\Delta_i^{i'}(g)$.
Now, for $i \in S_f^*$ and $i' \in T$ for some $T \in P_f \setminus \{S_f^*\}$ we get $h_i(f) - h_{i'}(f) = 1/(n-|T|)$ and
\begin{equation}\label{eq:my_final}
MY\Delta_i^{i'}(g^{\gamma}) = MY_i(g^{\gamma}) - MY_{i'}(g^{\gamma}) = \sum_{f \in C_k(G^{\gamma})} \frac{(-1)^{|P_f|} (|P_f| - 2)!}{\#_f \cdot (n - |T|)}.
\end{equation}
For a fixed coloring $f$, the weight of a coloring depends on $|P_f|$ and $T$, i.e., coalition in $P_f$ that contains $i'$.
However, if $i'$ is in a label of a node adjacent to all other nodes, then the weight depends solely on the number of colors $f$ uses.
In such a case, we can use a technique described at the beginning of this section (see \cref{eq:technique}).

\begin{theorem}\label{theorem:my_embedded}
For a game represented as embedded MC-nets or weighted MC-nets, computing the MY-value is \#P-hard.
\end{theorem}

\begin{figure}[t]
\centering
\begin{tikzpicture}[x=5cm,y=5cm] % change these values to adjust the size of a figure
  \def\x{0.0}

  \node[main node,label={below:$v_1$}] (1) at (\x+0.17, 1.00) {$1$}; 
  \node[main node,label={below:$v_2$}] (2) at (\x+0.00, 0.72) {$2$}; 
  \node[main node,label={below:$v_3$}] (3) at (\x+0.31, 0.72) {$3$}; 
  \node[main node] (4) at (\x+0.51, 1.00) {\ $\dots$}; 
  \node[main node,label={below:$v_k$}] (5) at (\x+0.63, 0.72) {$k$}; 
  \node[main node,label={below:$v_{k+1}$}] (6) at (\x+1.05, 0.85) {$k$+$1$};
  \node[main node,label={below:$v_{k+2}$}] (7) at (\x+1.30, 0.85) {$k$+$2$};
  \node[main node] (8) at (\x+1.55, 0.85) {\ $\dots$};
  \node[main node,label={below:$v_{k+j+1}$}] (9) at (\x+1.80, 0.85) {};
  \node[main node,label={below:$v^*$}] (10) at (\x+2.05, 0.85) {};

  \draw[->,thick] (\x+1.80,0.85) -- (\x+1.70,1.07);
  \node[draw=none] (6a) at (\x+1.62, 1.12) {$l(v_{k+j+1}) \!=\! k$+$j$+1};

  \draw[->,thick] (\x+2.05,0.85) -- (\x+2.15,0.98);
  \node[draw=none] (6a) at (\x+2.37,1.03) {$l(v^*)\!=\!k$+$j$+2, $\dots$, $3k$+1};

  \draw[dashed] (\x+0.33,0.80) ellipse (2.8cm and 1.8cm);
  \node[draw=none] at (\x+0.86,1.06) {$G$};

  \path[draw,thick]
  (4) edge  (3)
  (2) edge  (1)
  (1) edge[line width=2pt]  (3)
  (5) edge[line width=2pt]  (4)
  (5) edge  (1)
  ;

\end{tikzpicture}
\caption{Graph $\overline{G^{\gamma}}$ from the proof of \cref{theorem:my_embedded}.} 
\label{figure:my_embedded}
\end{figure}

\begin{proof}
%The value $MY_i(g^{\gamma})$ can be considered as the number of accepting paths of nondeterministic Turing machine, so the problem is in \#P.
%To show that the problem is \#P-hard, 
Again, we use a Turing reduction from the problem of counting all matching in a bipartite graph which is \#P-complete~\cite{Valiant:1979:enumeration}.

With the same reasoning as in the SS-value, instead of considering colorings in graph $G^{\gamma}$ we will focus on partitions of a graph into cliques in the complement graph $\overline{G^{\gamma}}$. Then, \cref{eq:my_final} can be rewritten as follows:
\begin{equation}\label{eq:my_final2}
MY\Delta_i^{i'}(g^{\gamma}) = \sum_{P \in QC(\overline{G^{\gamma}})} \frac{(-1)^{|P|} (|P| - 2)!}{(n - |T|)}.
\end{equation}

Let $G = (V,E)$ be an arbitrary bipartite graph with $V = \{v_1,\dots,v_k\}$.
Let $h_G^m$ be the number of all matchings in $G$ of size $k-m$; note that in such a case $m$ is the number of pairs plus the number of unmatched nodes.
We will determine $h_G^m$ for every $m \in \{1,\dots,k\}$.
The sum $\sum_{m=1}^{k} h_G^m$ is the number of all matchings in $G$.

To this end, let us construct $k$ player-graphs: for $j \in \{1,\dots,k\}$ we add $j+1$ new nodes $U = \{v_{k+1}, \dots, v_{k+j+1}\}$, label each node $v_i$ with $i$ and add yet another node $v^*$ with label of size $2k-j$.
Specifically, we construct a player-graph $\overline{G^{\gamma_j}} = (V \cup U \cup \{v^*\}, E, l, v^*)$ with $l(v_i) = i$ for every $v_i \in V \cup U$ and $l(v^*) = \{k+j+2,\dots,3k+1\}$.
Consider a complement of graph $\overline{G^{\gamma_j}}$, denoted by $G^{\gamma_j}$. 
Since node $v^*$ does not have any edges in $\overline{G^{\gamma_j}}$, then it is adjacent to all nodes in $G^{\gamma_j}$; hence, $G^{\gamma_j}$ is a regular player-graph, i.e., there exists a regular hybrid rule $\gamma_j$ equivalent to some embedded MC-nets rule such that $G^{\gamma_j}$ is the corresponding graph. 

Now, consider $MY\Delta_{i}^{i'}(g^{\gamma_j})$ for $i=k+j+2$ and $i'=k+j+1$.
Since node $v_{k+j+1}$ that contains $i'$ does not have any edges, we have $|T| = |l(v_{k+j+1})| = 1$.
From \cref{eq:my_final2} for game $g^{\gamma_j}$ we get that:
\[ MY\Delta_{i}^{i'}(g^{\gamma_j}) = \sum_{m=1}^k \frac{(-1)^{m+j+2} (m+j)!}{3k} \cdot h_G^m = \frac{1}{3k} \sum_{m=1}^k (-1)^{m+j} (m+j)! \cdot h_G^m. \]
This system of linear equations can be presented in the matrix form as follows:
\[
\begin{bmatrix}
2! & -3! & \dots  & \pm (k+1)! \\
-3! & 4! & \dots  & \mp (k+2)! \\
\vdots & \vdots & \ddots & \vdots \\
\pm (k+1)! & \mp (k+2)! & \dots  & (2k)!
\end{bmatrix}
\cdot 
\begin{bmatrix}
h_G^1 \\
h_G^2 \\
\vdots \\
h_G^k
\end{bmatrix}
=
\begin{bmatrix}
3k \cdot MY\Delta_i^{i'}(g^{\gamma_1}) \\
3k \cdot MY\Delta_i^{i'}(g^{\gamma_1}) \\
\vdots \\
3k \cdot MY\Delta_i^{i'}(g^{\gamma_1})
\end{bmatrix} 
\]
Note that by multiplying even rows and then even columns by $(-1)$ we can transform the square matrix into matrix $A = ((i+j)!)_{1 \le i,j \le k}$. 
Moreover, since there are as many even rows as even columns, the determinant of the original matrix is the same as the determinant of $A$.
Now, from \cite[Theorem 1.1]{Bacher:2002}, we know that the determinant of $A$ equals $\prod_{i=0}^{k-1} (i!)(i+2)!$; hence, the determinant of the original matrix is the same.
Since the determinant is non-zero, the matrix is invertible and knowing $MY\Delta_i^{i'}(g^{\gamma_1}), \dots, MY\Delta_i^{i'}(g^{\gamma_k})$ (i.e., $MY_i(g^{\gamma_1}), MY_{i'}(g^{\gamma_1}), \dots, MY_i(g^{\gamma_k}), MY_{i'}(g^{\gamma_k})$) allows us to find $h_G^1, \dots, h_G^k$ in polynomial time using Gaussian elimination. 
\end{proof}

%%%%%%%%%%%%%%%%%%%%%%%%%%%%%%%%%%%%%%%%%%%%%%%%%%%%%%%%%%%
%%%%%%%%%%%%%%      CONCLUSIONS              %%%%%%%%%%%%%%
%%%%%%%%%%%%%%%%%%%%%%%%%%%%%%%%%%%%%%%%%%%%%%%%%%%%%%%%%%%
\section{Conclusions}\label{section:conclusions}
In this paper, we studied the complexity of computing extended Shapley value in games represented as embedded and weighted MC-nets. 
Our results show that weighted MC-nets, which are more concise than embedded MC-nets, are slightly worse when it comes to the Shapley value computation.
For both representations only simplest extended Shapley values, MQ-value and EF-value, that ignore most values in a game can be computed in polynomial time (unless P=NP).
Also, combined with the work by \citeA{Skibski:etal:2020:pdt}, we get that computational properties of partition decision trees are significantly better than both MC-nets representations, as they allow the polynomial-time computation of five extended Shapley values.
This, however, comes at the cost of conciseness.

There are many possible directions for further research. 
It would be natural to study core-related questions for all three representations.
However, as in the case of the the Shapley value, there are multiple ways to extend the core to games with externalities which significantly increases the complexity of this task~\cite{Koczy:2018}.

Another interesting idea would be to analyze hybrid rules and corresponding player-graphs not as a representation, but as a graph-restriction scheme for games with externalities~\cite{Myerson:1977}.
In this way, considered extensions of the Shapley value can be interpreted as the extensions of the Myerson value for the analyzed setting.
Last but not least, a combination of games with externalities and graphs can result in new concepts of game-theoretic network centralities.
Our work can constitute a guideline which extensions should be used in order to obtain tractable measure.

\section*{Acknowledgments}
A preliminary version of the results in this paper appeared in the proceedings of the 34th AAAI Conference on Artificial Intelligence (AAAI-20). 
In this paper, we added proofs of Lemmas~\ref{lemma:hybrid_rules_weighted}--\ref{lemma:game_from_g_gamma} and full proofs (instead of sketches) of Theorems~\ref{theorem:mq_weighted}--\ref{theorem:my_embedded}.
We also added multiple examples, illustrations, intuitions and extended the related work section.
This work was supported by the Polish National Science Centre, under project 2015/19/D/ST6/03113.

\vskip 0.2in
\bibliography{bibliography}
\bibliographystyle{theapa}

\end{document}